\documentclass[letterpaper, 10 pt, journal, twoside, romanappendices]{IEEEtran}

\IEEEoverridecommandlockouts                              %

\usepackage{siunitx}
\usepackage{cite}
\usepackage{amsmath,amssymb,amsfonts, mathtools, amsthm}
\usepackage[table]{xcolor}
\usepackage[colorlinks,
  linkcolor=blue,
  citecolor=blue, urlcolor=.]{hyperref}
\usepackage{bm}
\usepackage{enumitem}
\usepackage[capitalise]{cleveref}
\usepackage{algorithm,algorithmic}
\newcommand{\crefdefpart}[2]{%
  \namecref{#1}~\hyperref[#2]{\labelcref*{#1}\ \ref*{#2}}%
}

\newtheorem{assumption}{Assumption}
\newtheorem{definition}{Definition}
\newtheorem{lemma}{Lemma}
\newtheorem{proposition}{Proposition}
\newtheorem{remark}{Remark}
\newtheorem{theorem}{Theorem}

\newtheorem{example}{Example}

\DeclareMathOperator*{\argmin}{arg\,min}

\DeclareMathOperator*{\epi}{epi}

\newcommand{\x}{\boldsymbol{x}}

\newcommand{\y}{\boldsymbol{y}}
\newcommand{\dd}{\mathop{}\!\mathrm{d}}
\newcommand{\inp}{\boldsymbol{u}}
\newcommand{\w}{\boldsymbol{w}}

\newcommand{\R}{\mathbb{R}}
\newcommand{\N}{\mathbb{N}}
\newcommand{\Prob}{\mathbb{P}}
\newcommand{\E}{\mathbb{E}}
\newcommand{\Q}{\mathbb{Q}}

\newcommand{\B}{\mathbb{B}}
\newcommand{\A}{\mathbf{A}}
\newcommand{\Emat}{\mathbf{E}}
\newcommand{\Z}{\mathcal{Z}}
\newcommand{\Pset}{\mathcal{P}(\mathcal{Z})}
\newcommand{\Bmat}{\mathbf{B}}
\newcommand{\Symm}{\mathbb{S}}

\makeatletter
\newcounter{relctr}
\everydisplay\expandafter{\the\everydisplay\setcounter{relctr}{0}}
\renewcommand*\therelctr{\alph{relctr}}
\AtBeginDocument{\let\originallabel\label} %

\newcommand\storelabelletter[1]{%
  \expandafter\xdef\csname savedletter@#1\endcsname{\therelctr}%
}

\newcommand\labelrel[2]{%
  \refstepcounter{relctr}%
  \storelabelletter{#2}%
  \stackrel{\textnormal{(\therelctr)}}{\mathstrut{#1}}%
  \originallabel{#2}%
}

\makeatother

\title{\LARGE \bf Sinkhorn Ambiguity Sets for Distributionally Robust Control: Convexity, Weak Compactness, and Tractability
}

\author{Riccardo Cescon, Andrea Martin, and Giancarlo Ferrari-Trecate
\thanks{R. Cescon and G. Ferrari-Trecate
are with the Institute of Mechanical Engineering, EPFL, Switzerland. E-mail addresses: \{riccardo.cescon, giancarlo.ferraritrecate\}@epfl.ch.}
\thanks{A. Martin is with the School of
Electrical Engineering and Computer Science, and Digital Futures, KTH Royal Institute of Technology, Sweden. E-mail address: andrmar@kth.se.}
\thanks{This work was supported as a part of NCCR Automation, a National Centre of Competence in Research, funded by the Swiss National Science Foundation (grant number 51NF40\_225155), and by Digital Futures.}
}

\begin{document}

\maketitle

\thispagestyle{empty}
\pagestyle{empty}
\begin{abstract}
Classical stochastic control assumes perfect knowledge of the uncertainty affecting the plant. In practice, however, such information is often incomplete. To address this limitation, we consider a distributionally robust control (DRC) problem with ambiguity sets defined via the Sinkhorn discrepancy. Compared to other discrepancy measures based on optimal transport, such as the popular Wasserstein distance, the Sinkhorn divergence does not constrain the worst-case distribution to be discrete, and allows combining observed data with prior knowledge in the form of a reference distribution, making this choice particularly suitable when only few noise samples are available for control design. We first study the properties of Sinkhorn ambiguity sets, establishing convexity and weak compactness under standard assumptions. We then leverage these results to prove that, the Sinkhorn DR linear quadratic control problem over linear policies can be solved through convex programming—even in the presence of DR safety constraints. Finally, we validate our theoretical findings and demonstrate the effectiveness of the proposed approach on a trajectory planning example.%

\end{abstract}

\section{Introduction}

Modern engineered systems must operate reliably in safety-critical applications despite uncertainty. 
Prominent examples include autonomous vehicles, electric power grids, and advanced robotic platforms. Traditional stochastic optimal control methods minimize an expected cost under the assumption that the true noise distribution is known \cite{bertsekas2012dynamic}. However, the properties of the uncertainty are often only indirectly observable through a finite number of samples. Therefore, such methods can exhibit unsatisfactory or brittle performance and constraint violations.
In contrast, robust control techniques \cite{zhou1998essentials} bypass a probabilistic modeling and treat uncertainty in an adversarial manner by optimizing for the worst-case scenario. These methods provably guarantee safety when the uncertainty is bounded, yet they tend to be excessively conservative as they safeguard against the least favorable disturbance realization. This conservatism leads to performance degradation in many applications. Motivated by this tradeoff, several approaches have been developed to balance robustness and average-case performance, including mixed $\mathcal{H}_2/\mathcal{H}_\infty$ formulations \cite{bernstein1988lqg, doyle1989optimal} and regret minimization frameworks \cite{martin2025guarantees, goel2023regret, martin2024regret}.
Within this landscape, distributionally robust optimization (DRO) has recently gained momentum in the control community. DRO aims to combine the flexibility of probabilistic modeling with the robustness of worst-case analysis by optimizing the expected performance under the most adverse distribution within a prescribed ambiguity set. Given a nominal probability distribution, ambiguity sets may be defined in various ways, such as through moment constraints \cite{van2015distributionally}, $\phi$-divergence neighborhoods \cite{falconi2025distributionally, petersen2000minimax, fochesato2025distributionally}, or optimal transport (OT) metrics—most notably, the Wasserstein distance \cite{mohajerin2018data}.

Several works have explored DRO-based control with Wasserstein ambiguity sets. 
By considering ambiguity sets centered at a nominal Gaussian distribution, \cite{taskesen2024distributionally, lanzetti2024optimality} prove global optimality of linear policies for DR LQG control. Other works employ data-driven ambiguity sets centered at the empirical distribution to benefit for greater expressiveness and capture multi-modal probabilities. For example, \cite{aolaritei2023wasserstein, coulson2021distributionally, micheli2022data} focus on DR model predictive control (MPC), ensuring recursive feasibility and chance constraint satisfaction. Instead, \cite{brouillon2025distributionally} considers a constrained stochastic control problem with stability guarantees, and assumes that the support of the true distribution lies in a compact polyhedron.

Despite their widespread adoption, Wasserstein-based DRO formulations have notable limitations. From a modeling standpoint, recent theoretical results \cite{gao2023distributionally} show that when the nominal distribution is discrete (e.g., empirical), the worst-case distribution for the Wasserstein problem is also discrete, which may fail to capture the true continuous nature of underlying noises. Moreover, while standard concentration bounds offer a principled way to calibrate the Wasserstein radius and ensure that the ambiguity set contains the true distribution with high probability, an excessively large radius may be required when data are scarce, potentially resulting in overly conservative controllers.

To address these challenges, we propose to model distributional uncertainty using the \emph{Sinkhorn discrepancy} \cite{cuturi2013sinkhorn}, an entropy-regularized variant of the optimal transport distance that has recently garnered attention in machine learning, optimization and control \cite{azizian2023regularization, blanchet2023unifying, dapogny2023entropy, cescon2025data, feng2026sinkhorn}. The Sinkhorn discrepancy offers several advantages over the Wasserstein metric: critically, it allows to shape the support of the worst-case distribution through the choice of a reference probability. This property allows for a richer and potentially more realistic description of uncertainty, as we illustrate in \cref{fig:wc_sinkhorn}. Additionally, the Sinkhorn discrepancy enables practitioners to integrate prior knowledge with empirical data by smoothly penalizing deviations from a reference law.

\textit{Contributions:} Motivated by the advantages of the Sinkhorn discrepancy, this article makes the following contributions. We first establish convexity and weak compactness of Sinkhorn ambiguity sets, a result of independent interest to the optimal transport literature. We then showcase how these properties allow rewriting the infinite-dimensional Sinkhorn DR conditional value-at-risk (CVaR) constraints as a convex finite-dimensional feasible region when the loss is piecewise affine. Our reformulation recovers known results for Wasserstein DRO when the entropic regularization parameter goes to zero. Last, by leveraging these results, we derive a tractable reformulation for the linear quadratic Sinkhorn DR control problem even in the presence of DR safety constraints. The effectiveness of the proposed control scheme is illustrated on a trajectory planning problem involving a Boeing B-747 aircraft operating under realistic wind turbulence models.

A preliminary version of this work has recently appeared in the 64th IEEE Conference on Decision and Control \cite{cescon2025data}. Differently from \cite{cescon2025data}, this paper establishes fundamental geometric and topological properties of Sinkhorn ambiguity sets, showing how these are key to solving Sinkhorn DR constrained control problems through convex programming. Furthermore, this work includes all the technical proofs. Last, new numerical experiments are presented to validate our control design method in the presence of probabilistic safety constraints and realistic noise models.

The remainder of this paper is organized as follows. \Cref{sec:preliminaries} reviews necessary background on DRO and optimal transport. In \cref{sec:distributional_uncertainty}, we characterize the topological and variational properties of the Sinkhorn ambiguity set. In \Cref{sec:optimal_control}, we showcase how the properties established in \Cref{sec:distributional_uncertainty} apply to reformulate the tractable control problem. Numerical experiments illustrating the method and its computational complexity appear in \cref{sec:numerical}.
\newline
\textit{Notation}. Throughout the paper, we denote the set of probability distributions supported on a measurable set $\mathcal{Z}$ by $\mathcal{P}(\mathcal{Z})$. We write $\mu \ll \nu$ to denote that a measure $\mu$ is absolutely continuous with respect to $\nu$. If $\mu, \nu$ are two measures, $\mu\times\nu$ represents the product measure.
For $n \in \mathbb{N}$, we write $[n]$ to denote the set $\{1,\dots, n\}$. Given two vectors $x,y\in\R^d$, we denote their element-wise multiplication with $x\odot y$.The space of all positive semidefinite matrices of size $d$ is denoted by $\mathbb{S}^d$. We denote by $\|\cdot\|$ the Euclidean norm. Given a positive definite matrix $A\in\R^{d\times d}$ and a vector $x\in\R^d$, we let $\|x\|_A = \sqrt{x^\top A x}$. The determinant of a square matrix $A$ is denoted by $|A|$. We denote the $j$-th row of $A$ by $A_j$ and the $j$-th component of the vector $a$ by $a_j$. Finally, the notation $\star^{\top}AB$ is short for $B^{\top}AB$.\looseness=-1

\section{Preliminaries}
\label{sec:preliminaries}
We begin by recalling definitions of discrepancies between probability distributions that will be used throughout the paper. Later in the section, we show in a simple example the expressivity power of the Sinkhorn discrepancy compared to the Wasserstein distance.
\begin{definition}[Transportation Cost Function, {\cite[Definition 2.14]{kuhn2025distributionally}}] A lower semicontinuous function
$c(x, y): \mathcal{Z} \times \mathcal{Z} \rightarrow \R_+$ that satisfies the identity of indiscernibles (i.e. $c(x,y) = 0$ if and only if $x=y$) is a transportation cost function.
\label{def:transport_cost}
\end{definition}
\begin{definition}[OT discrepancy, {\cite[Definition 2.15]{kuhn2025distributionally}}] The optimal transport discrepancy $OT_c: \Pset \times \Pset \rightarrow [0, +\infty]$ associated with any given transportation cost function $c$ is defined through 
\begin{equation}
\label{eq:OT}
    OT_c(\Prob, \Q) = \inf_{\gamma\in\Gamma(\Prob, \Q)}\E_\gamma[c(x, y)]\, ,
\end{equation}
where $\Gamma(\Prob, \Q)$ represents the set of all couplings $\gamma$ between $\Prob$ and $\Q$, that is, all joint probability distributions with marginals $\Prob$ and $\Q$.
\end{definition}
We introduce the Kullback-Leibler (KL) divergence, which will serve as the entropic regularizer in the Sinkhorn discrepancy.
\begin{definition}[KL divergence, {\cite[Definition 2.8]{kuhn2025distributionally}}]
Given $\Prob, \Q\in\mathcal{P}(\mathcal{Z})$ with $\Prob\ll \Q$, the KL divergence or relative entropy between $\Prob$ and $\Q$ is defined as
    \begin{equation*}
    \mathrm{KL}(\Prob\|\Q) = \E_\Prob\left[\log\left(\frac{\dd\Prob(x)}{\dd\Q(x)}\right)\right]\, .
\end{equation*}
\end{definition}
\begin{definition}[Sinkhorn discrepancy, {\cite[Definition 1]{sinkhorn}}]
\label{def:Sinkhorn}
Consider $\Prob, \Q\in\mathcal{P}(\mathcal{Z})$, and let $\mu, \nu$ be reference probability measures over $\mathcal{Z}$ 
such that $\Prob \ll \mu$ and $\Q \ll \nu$. For a given transport cost $c$ and regularization parameter $\epsilon \geq 0$ the Sinkhorn discrepancy 
between $\Prob$ and $\Q$ is defined as\looseness=-1
\begin{equation}
\label{eq:sinkhorn}
   W_c^{\epsilon}(\Prob, \Q) = \inf_{\gamma \in \Gamma(\Prob, \Q)}\left\{\E_\gamma [c(x, y)] + \\ 
   \epsilon \mathrm{KL}(\gamma|\mu\times\nu )\right\}\, .
\end{equation}
\end{definition}
As noted in \cite[Remark 2]{sinkhorn}, any choice of $\Prob\ll\mu$ in \eqref{eq:sinkhorn} is equivalent up to a constant. Since in DRO applications the center is known and fixed, 
without loss of generality, we will choose $\mu = \Prob$ in the subsequent analysis. Note also that the discrepancy $W_c^{0}$ coincides with the OT 
discrepancy in (\ref{eq:OT}). Consequently, if the cost function $c$ is defined as the $p$-th power of some metric in the space $\mathcal{Z}$ we retrieve 
the well-known $p$-th power of the Wasserstein distance between $\Prob$ and $\Q$ denoted with $W_p$, e.g. \cite{villani2009optimal}[Definition 6.1].

Intuitively, as the regularization parameter $\epsilon$ in \eqref{eq:sinkhorn} increases, the transport plan is increasingly penalized for deviating from the structure of the reference measure $\nu$.

Inherently, the quantities $OT_c$, $W_p$ and $W_c^{\epsilon}$ measure how different the probability distributions $\Prob$ and $\Q$ are. 
They also naturally provide a definition of uncertainty in the space of probability distributions. Specifically, we can denote the Sinkhorn ambiguity set, denoted as $\mathcal{S}$-set for brevity, of radius $\rho$ and centered at $\Prob$ by  
\begin{align}
\label{eq:ambiguity_set}
\B_{\rho, \epsilon}(\Prob) &\doteq \{ \Q\in\mathcal{P}(\mathcal{Z}):W_c^{\epsilon}(\Q, \Prob) \leq \rho \} \subset \mathcal{P}(\mathcal{Z})\, .
\end{align}
In words, $\mathbb{B}_{\rho, \epsilon}(\Prob)$ contains all probability distributions that are $\rho$-close to $\Prob$ in the Sinkhorn discrepancy.
Analogously, we can denote the OT and Wasserstein ambiguity sets with $\B_{\rho}(\Prob)$ and $\mathbb{W}_\rho(\Prob)$, respectively.

With the notion of ambiguity set, given a loss function $\ell: \R^d\rightarrow\R$, we can define the worst-case expected cost  as
\begin{equation}
\label{eq:worst-case risk}
\sup_{\Q\in\B_{\rho, \epsilon}(\Prob)}\ \E_{z\sim\Q}[\ell(z)]\, .
\end{equation}
It is also relevant to introduce the following strong dual formulation of \eqref{eq:worst-case risk} that will be used in \cref{sec:optimal_control}:
\begin{equation}
\label{eq:dual}
    \inf_{\lambda\geq0}\ \Biggl\{\! \lambda\rho + \lambda\epsilon\E_{x\sim\Prob}\biggl[\log\E_{z\sim\nu}\Bigl[e^{(\ell(z)-\lambda c(x,z))/(\lambda\epsilon)}\Bigr]\biggr]\!\Biggr\}\,,
\end{equation}
which was derived in \cite{sinkhorn} under the next assumption,
\begin{assumption}
The reference measure $\nu$, the transport cost $c(x,y)$, the function $\ell$, and the joint distribution $\gamma$ satisfy:
\label{assumption}
\begin{enumerate}[label=(\roman*)]
    \item $\nu\{z:0 \leq c(x, z) < \infty\} = 1$ for $\Prob$-almost every $x$;\label{assumption1}
    \item $\E_{z\sim\nu}\left[e^{-c(x, z)/\epsilon}\right] < \infty$ for $\Prob$-almost every $x$;\label{assumption2}
    \item the function $\ell$ is measurable;\label{assumption3}
    \item Every joint distribution $\gamma$ on $\mathcal{Z}\times\mathcal{Z}$ with first marginal distribution $\Prob$ has a regular conditional distribution\footnote{We refer to \cite[Chapter 5]{kallenberg2002foundations} for the concept of regular conditional distribution.} $\gamma_x$ given the value of the first marginal equals $x$.\label{assumption4}
\end{enumerate}
\end{assumption}
\begin{remark}
    We highlight that such assumptions hold in many concrete examples. \crefdefpart{assumption}{assumption1} ensures that almost every point of $\Prob$ can be transported to $\nu$. For example, this is true when considering $\nu$ to be Gaussian and the transport cost is infinite only on a set of points with the same cardinality of the rational numbers. \crefdefpart{assumption}{assumption2} is trivially satisfied if $\nu$ is a \textit{probability} measure as in our case. \crefdefpart{assumption}{assumption3} is fulfilled in several machine learning and control applications, for example when considering quadratic and linear losses. Finally, \crefdefpart{assumption}{assumption4} is always satisfied when $\Prob$ is empirical over a finite set of samples, which is the case considered in this paper. 
\end{remark}
We conclude this section clarifying the role of the reference measure $\nu$ in \eqref{eq:sinkhorn} and the consequent advantage of using the Sinkhorn discrepancy compared to Wasserstein. It has been shown that, when the center is the empirical distribution over $n$ points, the worst-case distribution for Wasserstein DRO is supported on at most $n+1$ points \cite{gao2023distributionally}. On the contrary, the assumption $\Q\ll\nu$ in \cref{def:Sinkhorn} implies that every distribution in the ambiguity set \eqref{eq:ambiguity_set}—and hence also the worst-case one—shares the same support as $\nu$. Therefore, when $\nu$ is any continuous measure in $\R^d$ and $\Prob$ is the empirical distribution, the worst-case distribution is supported in the whole $\R^d$. This is visually demonstrated in the following example.\looseness=-1
\begin{example}[Comparison of Wasserstein and Sinkhorn discrepancies]
Consider the quadratic loss function $\ell(z) = z^\top z$, $z\in\R^2$ and Euclidean transport cost, i.e. $c(x,y) = \|x-y\|^2$. Moreover, assume that the distribution $\Prob$ is discrete and supported on $x_1 = (0.25, 0.75)$ and $x_2 = (0.75, 0.25)$ and reference distribution $\nu\sim\mathcal{N}(0, I)$. From \cite[Remark 4]{sinkhorn}, given the optimal Lagrangian multiplier $\lambda^\star>0$ for the dual problem \eqref{eq:dual}, the Sinkhorn worst-case distribution can be written as 
\begin{equation*}
    \dd\Q_\star(z) = \E_{x\sim\Prob}\left[\alpha_x\exp\left((\ell(z) - \lambda^\star c(x, z))/(\lambda^\star\epsilon)\right)\right]\dd\nu(z)\, ,
\end{equation*}
where $\alpha_x = (\E_{z\sim\nu}[\exp((\ell(z) - \lambda^\star c(x, z))/(\lambda^\star\epsilon))])^{-1}$ is a constant ensuring that $\Q^\star$ is normalized to unity. Furthermore, in the case of quadratic loss function, we can compute the Wasserstein worst-case distribution solving the Quadratically Constrained Quadratic Program (QCQP) in \cite[Theorem 12]{kuhn2019wasserstein}. In \cref{fig:wc_sinkhorn}, we provide a visual comparison between these worst-case distributions. In the Sinkhorn case, the distribution shares the same support as $\nu$, making it continuous. In contrast, the Wasserstein worst-case distribution remains discrete, supported on just two points. If we know that the true underlying distribution is continuous then the Wasserstein case can lead to overly-conservative decisions because it does not exploit the distribution's support information. Additionally, as $\epsilon$ increases, we observe that the Sinkhorn worst-case distribution gradually shifts towards $\nu$, consistent with what we will see in point 3 of \cref{proposition: relationships}.
\end{example}
\begin{figure}
   \centering
   \includegraphics[trim=0 0 0 0, clip, width = \columnwidth]{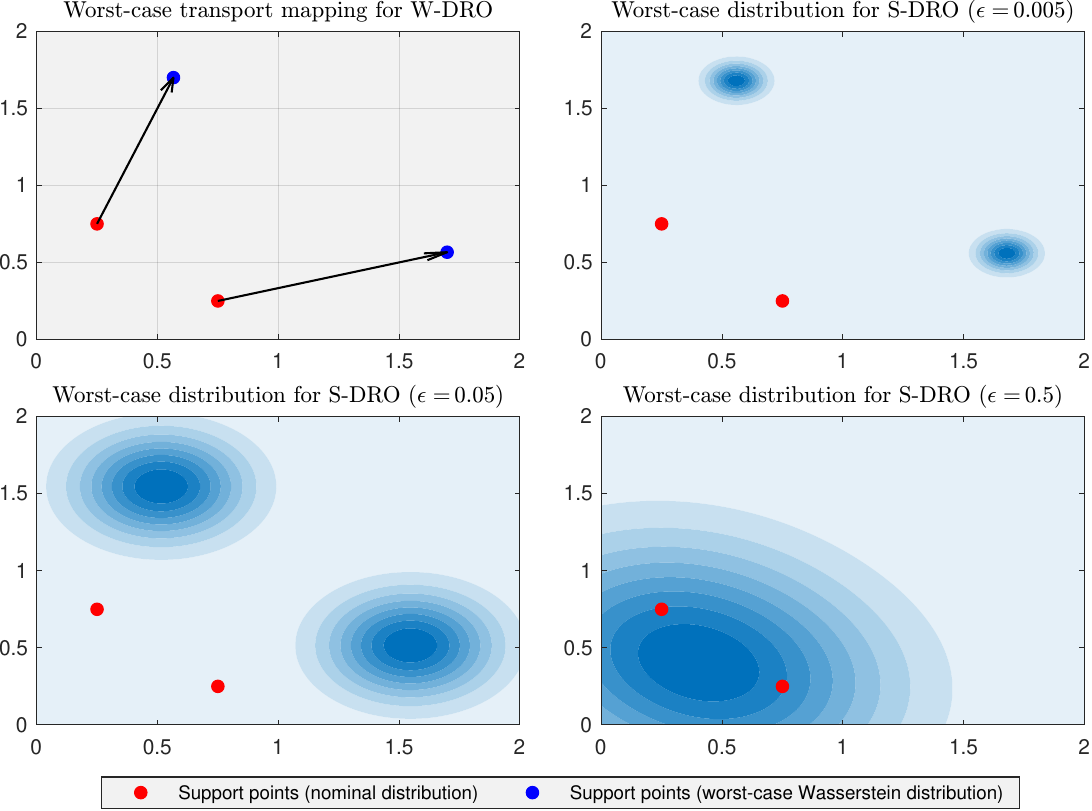}
    \caption{Visualization of worst-case distributions from Wasserstein DRO (top left plot) and Sinkhorn DRO (remaining three plots) with different choices of the regularization parameter $\epsilon$ in \eqref{eq:sinkhorn}. The red dots represent the supporting points of the empirical distribution. The blue dots in the Wasserstein case represent the supporting points of the worst-case distribution, while in the Sinkhorn plots the level sets of the worst-case distribution are depicted in different levels of blue.}
   \label{fig:wc_sinkhorn}
\end{figure}

\section{Properties of Sinkhorn ambiguity sets}
\label{sec:distributional_uncertainty}
The following section is devoted to the study of the $\mathcal{S}$-set. We first present some relationships between OT and $\mathcal{S}$-sets.
\begin{proposition}
\label{proposition: relationships}
Let $\Prob \in\mathcal{P}(\mathcal{Z})$ and fix the radius $\rho \geq 0$. Then, the following relationships hold:
\begin{enumerate}
\item $\B_{\rho, \epsilon}(\Prob) \subseteq \B_{\rho}(\Prob) \quad\forall \epsilon \geq 0$;
\item $\B_{\rho, \epsilon_2}(\Prob) \subseteq \B_{\rho, \epsilon_1}(\Prob) \quad \forall \epsilon_1, \epsilon_2 : 0\leq \epsilon_1 \leq \epsilon_2$;
\item If $\int_{\mathcal{Z}\times\mathcal{Z}} c(x, y)d\Prob(x) d\nu(y) \leq \rho$ then $\B_{\rho, \infty}(\Prob)$ is the singleton $\{\nu\}$ otherwise it is the empty set.
\end{enumerate}
\end{proposition}
The preceding proposition ensures that the $\mathcal{S}$-set is always contained within the OT ambiguity set. Furthermore, it establishes that the $\mathcal{S}$-sets are monotonically non-increasing (shrinking) with respect to the regularization parameter $\epsilon$. In the limit as $\epsilon \to \infty$, these sets converge to the singleton $\{\nu\}$, provided that the transport cost under the independent coupling $\mathbb{P} \times \nu$ does not exceed the budget $\rho$.

Convexity of the $\mathcal{S}$-set plays a key role in ensuring tractability of Sinkhorn DRO problems. This property is formally stated in the next Proposition whose
proof is deferred to Appendix \ref{appendix: convexity}.
\begin{proposition}[Convexity of the $\mathcal{S}$-set]
\label{proposition:convexity}
The set $\B_{\rho, \epsilon}(\Prob)$ in \eqref{eq:ambiguity_set} is convex.
\end{proposition}
\subsection{Topological properties}
Understanding the compactness of ambiguity sets is crucial in DRO, for example to determine the existence of worst-case distributions or establish that certain worst-case optimization problems admit Nash equilibria. This topological property has been addressed in \cite{kuhn2025distributionally} for various ambiguity sets. In this spirit, we proceed to investigate the compactness of the $\mathcal{S}$-set.\footnote{For clarity, we provide a discussion on relevant technical backgrounds on variational analysis and topology used in this section in Appendix \ref{appendix}.} We make the following assumption:
\begin{assumption}[Transportation cost]
\label{assum: transp_cost}
~ %
    \begin{enumerate}[label=(\roman*)]
        \item There exits a reference point $\hat{z}_0\in\mathcal{Z}$ such that, for some constant $C < \infty$, $\E_{z\sim\Prob}[c(z,\hat{z}_0 )] \leq C$ holds.\label{assum: c2}
        \item There exists a metric $d(z, \hat{z})$ on $\mathcal{Z}$ with compact sublevel sets such that $c(z, \hat{z}) \geq d(z, \hat{z})^p$ for some $p\in\N$.\label{assum: c3}
    \end{enumerate}
\end{assumption}
Such assumptions are not restrictive and commonly adopted in the literature, see e.g. \cite{shafieezadeh2023new}, \cite{yue}. 
Particularly, Assumption \ref{assum: transp_cost}\ref{assum: c3} allows for transportation costs that fail to display common properties such as symmetry, 
convexity, homogeneity, or the triangle inequality.

By leveraging the lower semicontinuity of the KL regularization term, our next result establishes the solvability of \eqref{eq:sinkhorn} by extending the proof given in \cite{kuhn2025distributionally} for the Wasserstein distance. The detailed proof is given in Appendix \ref{appendix: solvability}.
\begin{lemma}[Solvability of the Sinkhorn problem]The infimum in the Sinkhorn problem \eqref{eq:sinkhorn} is attained. 
\label{lemma:existance}
\end{lemma}
Building upon the previous lemma, we can generalize the continuity results for Wasserstein showing that, even with the additional entropic regularization term, the Sinkhorn discrepancy constitutes a weakly lower semicontinuous mapping of its arguments. This is stated in the following lemma, whose proof is deferred to Appendix \ref{appendix: semicontinuity}.
\begin{lemma}[Weak Lower Semicontinuity of the Sinkhorn Optimal Transport Discrepancy] The Sinkhorn Optimal Transport Discrepancy $W_c^{\epsilon}(\Prob, \Q)$ in \eqref{eq:sinkhorn} is weakly lower semicontinous jointly in $\Prob$ and $\Q$ for every $\epsilon\geq 0$.
\label{lemma: semicontinuity}
\end{lemma}
The lower semicontinuity of the Sinkhorn discrepancy is a key enabler to show that the $\mathcal{S}$-set is closed and hence compact. This result is stated formally in the Theorem below, whose proof is deferred to Appendix \ref{appendix:theorem_compactness}.
\begin{theorem}
    \label{thm: Sinkhorn_compactness}
    If the transportation cost $c$ satisfies Assumption \ref{assum: transp_cost}, then the ambiguity set $\B_{\rho, \epsilon}(\Prob)$ in \eqref{eq:ambiguity_set} is weakly compact.
\end{theorem}
The compactness of the $\mathcal{S}$-set, which constitutes a new result in the literature, serves as a stepping stone condition for establishing min-max results. In the following, we will use it in \cref{proposition: DR_Sinkhorn_CVaR} to derive distributionally robust CVaR constraints with uncertainty belonging to an $\mathcal{S}$-set.
\section{Finite Horizon Sinkhorn Robust Constrained Optimal Control}
\label{sec:optimal_control}
In this section, we analyze the finite-horizon
Sinkhorn robust optimal control problem with probabilistic constraints. We begin by describing the setup, then reformulate the Sinkhorn robust CVaR constraints as a convex set using the topological results from the previous section, and conclude by deriving a tractable convex optimization program that solves the control problem.
\subsection{Problem setup}
Consider the discrete-time linear time-varying system described by the state-space model
\begin{equation}
\label{eq:system}
x_{t+1} = A_t x_t + B_t u_t + E_t w_t,\quad
y_t = C_t x_t + F_t w_t\, ,
\end{equation}
where the variables are defined as follows: $x_t \in \R^d$ is the system state, $u_t \in \R^m$ the control input, $w_t \in \R^q$ a stochastic disturbance, and $y_t \in \R^p$ the measured output.

The state and input are constrained to lie in polytopic feasible sets $\mathcal{X}\subseteq \R^d$ and $\mathcal{U}\subseteq \R^m$, respectively, given by
\begin{align*}
    \mathcal{X} &= \{x\in\R^d: h(x) = \max_{j\in[L_x]}\ H_j^{\top}x + h_j \leq 0, L_x\in\N\}\, ,\\
    \mathcal{U} &= \{u\in\R^m: g(u) = \max_{j\in[L_u]}\ G_j^{\top}u + g_j \leq 0, L_u\in\N\}\, .
\end{align*}
We consider the problem of designing a feedback policy that retains probabilistic safety while ensuring performance guarantees over a finite horizon of length $T\in\N$.
Specifically, given a cost matrix $D\succeq 0$, we measure the control cost that a policy $u_t = \pi_t(y_{0:t})$ incurs in response to the disturbance 
realization $\w = (x_0^\top, w_0^\top, \dots, w_{T-2}^\top)^\top$ by
\begin{equation*}
    J(\bm{\pi}, \w) = \sum_{t=0}^{T-1}\begin{bmatrix}x_t^{\top} & u_t^{\top}\end{bmatrix}D\begin{bmatrix}x_t\\u_t\end{bmatrix}\, .
\end{equation*}
Since the noise is unknown and unbounded, enforcing strict constraints is infeasible. Instead, we employ a probabilistic constraint formulation, requiring the system to remain within bounds with a predefined confidence. Formally, if $\Prob$ is the true noise distribution,
we require that the random vectors $\x = (x_0^\top, \dots, x_{T-1}^\top)^\top\ \text{and}\ \inp = (u_0^\top, \dots, u_{T-1}^\top)^\top$ are contained in the feasible sets with probability at least $1-\gamma$, i.e. $\Prob(x_t \in \mathcal{X}_t,\, u_t \in \mathcal{U}_t) \ge 1 - \gamma, \quad \forall t \in \{0, \dots, T-1\}$
where $\gamma\in[0, 1]$ is an user-defined safety parameter that controls the level of acceptable constraint violation.

However, enforcing such constraints is difficult because they give rise to a non-convex feasible region in the variable space. To address this issue, an alternative commonly employed, convex, risk measure is the conditional value-at-risk (CVaR) for which a pictorial representation is given in \cref{fig:CVaR}.
\begin{definition}[Conditional value-at-risk, \cite{rockafellar2000optimization}] For any measurable loss function $\ell: \R^d\rightarrow\R$, probability distribution $\Prob \in \Pset$ and 
tolerance $\gamma\in[0, 1]$, the CVaR of the random loss $\ell(\x)$ at level $\gamma$ with respect to $\Prob$ is defined as
\begin{equation}
    \label{eq:CVaR}
        \text{CVaR}_{\gamma}^\Prob(\ell(\x)) = \inf_{\tau\in\R}\ \tau + \frac{1}{\gamma}\E_\Prob[\max\{\ell(\x) - \tau, 0\}]\,.
    \end{equation}
\end{definition}
We highlight that, besides implying the constraint satisfaction in probability, \eqref{eq:CVaR} accounts also for the 
expected amount of constraint violation in the $\gamma$ percent of cases when such violations occur \cite{van2015distributionally}. This is relevant 
for control applications where severe breaches of the safety constraints often have far more adverse consequences than mild violations.
\begin{figure}[t]
    \centering
    \includegraphics[trim=0 0 0 0, clip, width = \columnwidth,height=0.15\textheight,keepaspectratio]{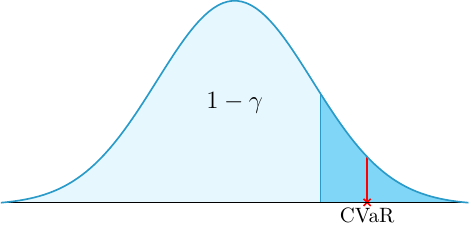}
    \caption{CVaR is the expected value of the $\gamma$-right tail (dark blue).}
    \label{fig:CVaR}
\end{figure}
The constrained control problem can then be expressed as
\begin{subequations}
    \label{eq:trueP_problem}
\begin{align}
    &\boldsymbol{\pi}^{\star} \in \argmin_{\boldsymbol{\pi} = (\pi_0, \dots, \pi_{T-1})}\ \E_{\Prob}[ J(\mathbb{\bm\pi}, \w)]\label{eq:nominal_objective}\\
    & \text{subject to}\ \text{CVaR}_{\gamma}^\Prob(\max\{h(\x), g(\inp)\}) \leq 0\,.
    \label{eq:nominal_constraints}
\end{align}
\end{subequations}
\subsection{Approximations and problem formulation}
Problem \eqref{eq:trueP_problem} is rarely solvable directly. Indeed, most of the time, the true underlying probability $\Prob$ of $\w$ is unknown. Instead, we assume to have access to $N\in\N$
independent noise realizations $\hat{\w}_T^1, \dots, \hat{\w}_T^N$, where each sample 
\begin{equation*}
    \hat{\w}_T^i = \begin{bmatrix}(\hat{w}_0^i)^{\top}\ (\hat{w}_1^i)^{\top}\ \dots\ (\hat{w}_{T-1}^i)^{\top}\end{bmatrix}^{\top}
\end{equation*}
constitutes a trajectory of length $T$ of $w_t$. Therefore, we consider the empirical distribution constructed with the noise trajectories, 
$\hat{\Prob} = \frac{1}{N}\sum_{i=1}^{N}\delta_{\hat{\w}_T^i}$,
where $\delta_{\hat{\w}_T^i}$ is the Dirac delta centered at $\hat{\w}_T^i$.

Naively solving (\ref{eq:trueP_problem}) replacing the true distribution $\Prob$ with $\hat{\Prob}$ may lead to decisions that are unsafe or with poor 
out-of-sample performance. To combat this phenomenon, we consider a robust counterpart of the nominal objective \eqref{eq:nominal_objective} and CVaR constraint \eqref{eq:nominal_constraints} that 
minimizes the worst-case expected cost with worst-case CVaR constraints over the $\mathcal{S}$-set. This allows us to bypass the issue of not knowing the true distribution.

However, even with such assumption, solving \eqref{eq:trueP_problem} for general feedback policies remains challenging. In principle, dynamic programming could be applied, but it is computationally intractable in continuous state spaces. Motivated by the recent work \cite{cescon2025global} on Sinkhorn DR control, in which the authors showed the optimality of linear policies when the nominal distribution is Gaussian, we
restrict ourselves to policies that are affine in the past observations $\y = (y_0^\top, \dots, y_{T-1}^\top)^\top$, i.e. 
\begin{equation}
\label{eq:feedback law}
    u_t = \sum_{k=0}^t K_{t,k}y_k + v_t,\ \forall t\in\{0, \dots, T-1\}\,,
\end{equation} or more compactly
\begin{equation}
\label{eq:compact_feedback_law}
    \inp = \mathbf{K}\y + \mathbf{v}\,,
\end{equation}
where $\mathbf{K}\in\R^{mT\times pT}$ is a real lower block-triangular feedback matrix, while $\mathbf{ v}\in\R^{mT}$ is an affine term.

To ensure the tractability of the resulting optimization, we adopt the System Level Synthesis (SLS) framework, shifting the control design from the affine output-feedback parameters $(\mathbf{K, v})$ to the closed-loop system response maps $(\mathbf{\Phi}, \bm\phi)$. Under this parametrization, the input-output variables are characterized by linear maps of the stacked disturbance $\Tilde{\w}$, and the performance objective $J$ admits a convex representation. The explicit derivation of these closed-loop maps and the equivalence between the controller structures are detailed in Appendix~\ref{appendix:SLS}.

With these definitions in place, the Sinkhorn distributionally robust (DR) constrained optimal control problem is formulated as follows:
\begin{subequations}
\label{eq:Sinkhorn_problem}
\begin{align}
    &(\boldsymbol{\Phi}^{\star}, \bm\phi^\star) \in \argmin_{\boldsymbol{\Phi}, \bm\phi}\ \sup_{\Q\in\B_{\rho, \epsilon}(\hat{\Prob})}\ \E_{\Q}[ J(\mathbf{\Phi}, \bm\phi,\Tilde{\w})]\label{eq:Sinkhorn_cost}\\
    & \text{s.t.}\ \sup_{\Q\in\B_{\rho, \epsilon}(\hat{\Prob})}\ \text{CVaR}_{\gamma}^\Q(M_j(\mathbf{\Phi}\tilde{\w} + \bm\phi) +m_j)\label{eq:Sinkhorn_CVaR} \leq 0\quad \forall j\in[\tilde{L}]\,,
\end{align}
\end{subequations}
where $\tilde{L} = L_x+L_u$, $M_j = \text{diag}(\begin{bmatrix}
    H, G
\end{bmatrix})_j$, and $m_j = \begin{bmatrix}
    h^\top, g^\top
\end{bmatrix}_j$.

It is important to note that the worst-case distributions in \eqref{eq:Sinkhorn_cost} and \eqref{eq:Sinkhorn_CVaR} do not necessarily coincide. In practice, since the true uncertainty distribution is unique, problem \eqref{eq:Sinkhorn_problem} may lead to a conservative solution. Nonetheless, robustness must be ensured for all distributions within the ambiguity set $\B_{\rho, \epsilon}(\hat{\Prob})$, rather than only for the one that minimizes the expected worst-case cost.

In our formulation, we considered a single ambiguity set characterized by the parameters $\rho$ and $\epsilon$. With minor modifications, this framework could be extended to allow distinct values of these parameters for the cost and the constraint, thereby enabling different levels of conservatism between \eqref{eq:Sinkhorn_cost} and \eqref{eq:Sinkhorn_CVaR}.

In the sequel, when defining the $\mathcal{S}$-set $\B_{\rho, \epsilon}(\hat{\Prob})$ we will make the following assumption
\begin{assumption}
The reference measure $\nu$ in \eqref{eq:sinkhorn} is a multidimensional Gaussian distribution\footnote{Our results can in principle be extended to any reference distribution $\nu$ satisfying \crefdefpart{assumption}{assumption1}-\ref{assumption2}. Exact reformulations of \eqref{eq:Sinkhorn_problem}, however, require case-by-case computations.} with mean vector $m\in\R^s$ and covariance matrix $\Sigma \in\mathbb{S}^s$, that is
\begin{equation*}
\dd\nu(\xi) \! = \! C_s^{-1}\exp\left(\!-\frac{1}{2}(\xi-m)\!^{\top}\Sigma^{-1}(\xi-m)\!\right)\dd\lambda^s(\xi)
\end{equation*}
with $\lambda^s$ the $s$-dimensional Lebesgue measure and $C_s = \sqrt{(2\pi)^s|\Sigma|}$.
\end{assumption}
\subsection{Sinkhorn Distributionally Robust CVaR reformulation}
We study Sinkhorn DR CVaR constraints that will enable us to derive a strong dual reformulation for \eqref{eq:Sinkhorn_CVaR}. 

We consider the CVaR constraint in \eqref{eq:CVaR} when the true probability $\Prob$ is unknown. In fact, it is taken from the $\mathcal{S}$-set
as the one returning the worst-case CVaR. Moreover, the loss function $\ell$ considered is the maximum of $J\in\N$ linear pieces because input and state are required to lie within polytopes. 

The following Proposition characterizes the feasible set for the variable satisfying the DR CVaR constraint when the ambiguity set is constructed using the 
Sinkhorn discrepancy. The proof of the Proposition leverages the convexity and weak compactness established in the previous section to interchange the $\sup$ and $\inf$ that appears in the definition of worst-case CVaR. The details are deferred to Appendix \ref{appendix: DR_Sinkhorn_CVaR}.
\begin{proposition}
\label{proposition: DR_Sinkhorn_CVaR}
Assume that $\mathcal{Z} = \R^d$. Consider the following DR CVaR constraint
\begin{equation}
    \label{eq:DR_Sinkhorn_CVaR}
\sup_{\Q\in\B_{\rho,\epsilon}(\hat{\Prob}_n)}\ \text{CVaR}_\gamma^\Q\left( \max_{j\in[J]}\ a_j^{\top}x + b_j\right) \leq 0,
\end{equation}
where $\B_{\rho,\epsilon}(\hat{\Prob}_n)$ is the Sinkhorn ambiguity set centered at the empirical distribution $\hat{\Prob}_n = 1/n\sum_{i=1}^n\delta_{\hat{\xi}_i}$ supported on the datapoints $\{\hat{\xi}_i\}_{i=1}^n$, with transport cost $c(x, y) = \|x -y\|^2$.
If
\begin{align}
\label{eq:feasibility_CVaR}
\rho \geq &\frac{\epsilon}{2} \log\left|\Sigma + \frac{\epsilon}{2}I\right| - \frac{\epsilon d}{2}\log\left(\frac{\epsilon}{2}\right) + \frac{\epsilon}{2}\|m\|^2_{\Sigma^{-1}} + \|\hat{\xi}_i\|^2
\nonumber\\
-& \frac{1}{n} \sum_{i=1}^n \star^{\top}\left(I +\frac{\epsilon}{2}\Sigma^{-1}\right)^{-1}\left(\hat{\xi}_i + \frac{\epsilon}{2}\Sigma^{-1}m\right)\,,
\end{align}
then \eqref{eq:DR_Sinkhorn_CVaR} is equivalent to the following set of constraints,
\begin{align}
\label{eq:CVaR_set}
&\forall i \in [n], \forall j\in[J+1]:\nonumber\\
&\left\{
    \begin{aligned}
    & \tau\in\R,\ \sigma\in\R_+,\ s_i\in\R\\ 
    &\sigma \rho + \frac{\gamma-1}{\gamma}\tau + \frac{1}{n} \sum_{i=1}^n s_i \leq 0\\
    &\setlength\arraycolsep{1pt}\begin{bmatrix}
        4\sigma I + \sigma\epsilon\Sigma^{-1} &\frac{a_j}{\gamma} + 2\sigma \hat{\xi}_i + \sigma\epsilon\Sigma^{-1}m \\
        \star & s_i - \zeta - \frac{b_j}{\gamma} + \sigma\|\hat{\xi}_i\|^2 + \frac{\sigma\epsilon}{2}\|m\|^2_{\Sigma^{-1}}
    \end{bmatrix}\! \succeq\! 0\,,
    \end{aligned}
\right.
\end{align}
with $a_{J+1}\!=\!0$, $b_{J+1}\!=\!\tau$, and $\zeta = \frac{\sigma\epsilon d}{2}\log\epsilon - \frac{\sigma\epsilon}{2}\log|2\Sigma + \epsilon I|$.
\end{proposition}
\begin{remark}
\label{remark: Wasserstein_CVaR}
As expected, we show in Appendix \ref{appendix:remark_proof} that, when $\epsilon\rightarrow 0$, the set of constraints in \eqref{eq:CVaR_set} reduces to the one obtainable with the Wasserstein distance.
\end{remark}
By virtue of \cref{proposition: relationships}, the $\mathcal{S}$-set is contained within the Wasserstein set. Consequently, satisfaction of the worst-case Wasserstein CVaR constraints implies satisfaction of the Sinkhorn constraints. This follows immediately from the fact that the supremum of a function over a subset is upper-bounded by the supremum over the enclosing set.
\subsection{Tractable reformulation of the control problem}
\label{subsection:reformulation}
In this subsection, we show that \eqref{eq:Sinkhorn_problem} can be recast into a convex optimization problem. The first Theorem characterizes the optimization problem to find the optimal maps in \eqref{eq:Sinkhorn_problem}. Instead, the subsequent Proposition proves that the problem is actually convex.
\begin{theorem}
\label{thm:main}
Assume $\mathcal{Z} = \R^s$, $c(x, y) = \|x -y\|^2$, and let $s = d + (T-1)p$. Problem \eqref{eq:Sinkhorn_problem} is feasible if and only if
\begin{align}
\rho \geq &\frac{\epsilon}{2} \log\left|\Sigma + \frac{\epsilon}{2}I\right| - \frac{\epsilon s}{2}\log\left(\frac{\epsilon}{2}\right) + \frac{\epsilon}{2}\|m\|^2_{\Sigma^{-1}} + \|\hat{\w}^i_T\|^2
\nonumber\\
-& \frac{1}{N} \sum_{i=1}^N \star^{\top}\left(I +\frac{\epsilon}{2}\Sigma^{-1}\right)^{-1}\left(\hat{\w}^i_T + \frac{\epsilon}{2}\Sigma^{-1}m\right).\label{eq:feasibility_condition_control}
\end{align}
Moreover, the optimal closed-loop maps $\mathbf{\Phi}^{\star}, \bm\phi^\star$ in \eqref{eq:Sinkhorn_problem} are given by the minimizer of
\begin{subequations}
\label{eq:convex SLS}
\begin{align}
&\inf\ \lambda\rho + \frac{1}{N}\sum_{i=1}^N s_i
\nonumber
\\
&\textrm{subject to},\ \forall i\in[N],\ \forall j \in[\Tilde{L}]:\nonumber
\\
&\lambda\in\R_+,\ Q\in \Symm^s,\ q\in\R^s,\ c\in\R,\  s_i\in\R,\ \zeta_i\in\R,\nonumber
\\
&\tau\in\R,\ \sigma\in\R_+,\ z_i\in\R\nonumber
\\
& M = \lambda\left(I+\frac{\epsilon}{2}\Sigma^{-1}\right) - Q, M \succ 0\label{eq:M>0}
\\
&\frac{\lambda\epsilon s}{2}\log\left(\frac{\lambda\epsilon}{2}\right) - \frac{\lambda\epsilon}{2}\log|\Sigma|- \frac{\lambda\epsilon}{2}\log|M| + \zeta_i\leq s_i\label{eq:logdet inequality}
\\
& \begin{bmatrix}
M && q + \lambda (\hat{\w}^i_T + \frac{\epsilon}{2}\Sigma^{-1}m)\\
\star && \zeta_i - c + \lambda\|\hat{\w}^i_T\|^2 + \frac{\lambda\epsilon}{2}\|m\|^2_{\Sigma^{-1}}\end{bmatrix}\succeq 0\label{eq:LMI2}\\
& \begin{bmatrix}
\label{eq:LMI1 Schur}
\begin{bmatrix}
    Q & q \\ \star &c
\end{bmatrix} && \star\\[1em]
D^\frac{1}{2}\begin{bmatrix}
  \mathbf{\Phi} & \bm\phi  
\end{bmatrix} && I
\end{bmatrix}\succeq 0
\\
&\sigma \rho + \frac{\gamma-1}{\gamma}\tau + \frac{1}{N} \sum_{i=1}^N z_i \leq 0\label{eq:CVaR1}
\\
&\setlength\arraycolsep{1pt}\begin{bmatrix}
        4\sigma I + \sigma\epsilon\Sigma^{-1} &\frac{1}{\gamma}M_j\mathbf{\Phi} + 2\sigma \hat{\w}^i_T + \sigma\epsilon\Sigma^{-1}m \\
        \star & z_i - \zeta - \frac{1}{\gamma}(m_j + M_j\bm\phi) + \sigma\|\hat{\w}^i_T\|^2
    \end{bmatrix}\! \succeq\! 0\label{eq:CVaR2}
\\
& \ \bm{\Phi}\in\R^{T(d+m)\times s} , \bm\phi\in\R^{T(d+m)}\ \textrm{satisfying}\ \eqref{eq:affine_subspace}\nonumber\,,
\end{align}
\end{subequations}
with $a_{\Tilde{L}+1}\!=\!0$, $b_{\Tilde{L}+1}\!=\!\tau$, and $\zeta = \frac{\sigma\epsilon d}{2}\log\epsilon - \frac{\sigma\epsilon}{2}\log|2\Sigma + \epsilon I| - \frac{\sigma\epsilon}{2}\|m\|^2_{\Sigma^{-1}}$.
\end{theorem}
\begin{proof}
Consider the inner supremum in \eqref{eq:Sinkhorn_cost}. We apply \cref{thm: duality} in Appendix \ref{appendix:lemma} with $Q = \mathbf{\Phi}^{\top} D\mathbf{\Phi}$ and $q = \mathbf{\Phi}^\top D\bm\phi$. The feasibility condition \eqref{eq:feasibility_condition_control} was already proved in the Lemma. By performing a change of variables, the affine term $\bm\phi^\top D\bm\phi$ can be shifted from the objective to the left hand side of the constraint in (\ref{eq:convex program}). Next, we introduce auxiliary variables $\zeta_i$, for $i\in[N]$ to upper bound the linear part of the constraint in (\ref{eq:convex program}) which directly yields the constraint \eqref{eq:logdet inequality}. Now, consider the remaining linear part of the constraint
\begin{align*}
\star^{\top}\left(\lambda\left(I +\frac{\epsilon}{2}\Sigma^{-1}\right)-Q\right)^{-1}\left(q + \lambda\left(\hat{\w}^i_T+\frac{\epsilon}{2}\Sigma^{-1}m\right)\right)+
\\
\bm\phi^\top D\bm\phi- \lambda \|\hat{\w}^i_T\|^2 - \frac{\lambda\epsilon}{2}\|m\|^2_{\Sigma^{-1}}\leq \zeta_i\label{eq:second_ineq}\,.
\end{align*}
With the change of variables $c = \bm\phi^\top D\bm\phi$ and applying the Schur complement, this constraint is equivalent to \eqref{eq:LMI2}. 
Moreover, applying \cref{proposition: DR_Sinkhorn_CVaR} to \eqref{eq:Sinkhorn_CVaR} yields \eqref{eq:CVaR1}-\eqref{eq:CVaR2}.\looseness=-1

Next, define $\Tilde{\mathbf{\Phi}} =\begin{bmatrix}
    \mathbf{\Phi} & \bm\phi
\end{bmatrix}$ and $S = \begin{bmatrix}
    Q & q \\ \star &c
\end{bmatrix}$. We must enforce $S = \Tilde{\mathbf{\Phi}}^\top D \Tilde{\mathbf{\Phi}}$, which is a non-convex constraint. Nevertheless, we relax it to $S \succeq \Tilde{\mathbf{\Phi}}^\top D \Tilde{\mathbf{\Phi}}$, which can be reformulated with Schur complement as \eqref{eq:LMI1 Schur}. This leads to problem \eqref{eq:convex SLS}, denoted with $P_{\text{ineq}}$. Similarly, we define $P_{\text{eq}}$ as the same problem but imposing $S \preceq \Tilde{\mathbf{\Phi}}^\top D \Tilde{\mathbf{\Phi}}$. Note that $P_{\text{ineq}}$ can have a smaller optimal solution because we enlarged the feasible set compared to $P_{\text{eq}}$. We now show that this relaxation is tight.

Let $\{Q^\star, q^\star, c^\star, \mathbf{\Phi}^\star, \bm\phi^\star, \lambda^\star, s^\star, \zeta^\star, \tau^\star, \sigma^\star, z^\star\}$ be an optimal solution to \eqref{eq:convex SLS} with optimal cost $J^{\star} = \lambda^{\star}\rho + \frac{1}{N}\sum_{i=1}^N s^{\star}_i$. Moreover, we construct a candidate solution for $P_{\text{eq}}$ as $\{\Tilde{S},\mathbf{\Phi}^{\star}, \bm\phi^\star, \lambda^{\star}, s^{\star}, \zeta^{\star}, \tau^\star, \sigma^\star, z^\star\}$, with $\Tilde{S} = \begin{bmatrix}
    \mathbf{\Phi}^\star & \bm\phi^\star
\end{bmatrix}^\top D\begin{bmatrix}
    \mathbf{\Phi}^\star & \bm\phi^\star
\end{bmatrix}$. This satisfies the equality constraint by design and it has the same cost $J^{\star}$ since the latter depends only on $\lambda^{\star}$ and $s^{\star}$. Thus, if such a candidate point is feasible, it is also optimal. Hence, we verify that the constraints in \eqref{eq:convex SLS} are satisfied. 

We start with \eqref{eq:M>0} by noticing that $\Tilde{S} \preceq S^\star$ holds by construction. Moreover, since the principal minors of a positive semidefinite matrix are nonnegative we have $\mathbf{\Phi}^{\star \top} D \mathbf{\Phi}^\star\preceq Q^\star$. This implies that \eqref{eq:M>0} is satisfied because $\lambda \left(I + \frac{\epsilon}{2}\Sigma^{-1}\right) - \mathbf{\Phi}^{\star \top} D \mathbf{\Phi}^\star \succeq \lambda \left(I + \frac{\epsilon}{2}\Sigma^{-1}\right) - Q^{\star} \succ 0$. Furthermore, \eqref{eq:logdet inequality} is satisfied because it holds
\begin{equation*}
\log\left|\lambda \left(I + \frac{\epsilon}{2}\Sigma^{-1}\!\right) - Q^{\star}\right|\! \leq\! \log\left|\lambda \left(I + \frac{\epsilon}{2}\Sigma^{-1}\!\right) - \mathbf{\Phi}^\star D \mathbf{\Phi}^\star\right|\, ,
\end{equation*}
where we used the monotonicity of the logarithm function together with the fact that $A \preceq B$ implies $|A| \leq |B|$. Finally, we show that the candidate solution satisfies \eqref{eq:LMI2}. The constraint can be written as
\begin{align*}
&\setlength\arraycolsep{.9pt}\underbrace{\begin{bmatrix}
\lambda^\star\left(I+\frac{\epsilon}{2}\Sigma^{-1}\right) - Q^\star & q^\star + \lambda^\star( \hat{\w}^i_T + \frac{\epsilon}{2}\Sigma^{-1}m)
\\
\star & \zeta_i - c^\star + \lambda^\star\|\hat{\w}^i_T\|^2 + \frac{\lambda^\star\epsilon}{2}\|m\|^2_{\Sigma^{-1}}
\end{bmatrix}}_{\succeq 0}\! +\!
\\
&\underbrace{\begin{bmatrix}
Q^\star - \mathbf{\Phi}^\star D \mathbf{\Phi}^\star & {\mathbf{\Phi}^\star}^\top D \phi^\star - q^\star
\\
\star & c^\star - \phi^\star D\phi^\star
\end{bmatrix}}_{(\Delta)}\, ,
\end{align*}
where the first term is positive semidefinite because the optimal solution satisfies \eqref{eq:LMI2} by construction. We notice that $(\Delta) = P^\top (S^\star - \Tilde{S})P$, with $P = \begin{bmatrix}
    I & 0\\
    0 & -1
\end{bmatrix}$. Since $P$ is nonsingular, thanks to Sylvester's law of inertia, $(\Delta)$ has the same number of positive, negative and null eigenvalues of $S^\star - \Tilde{S}$. This in turn implies that $(\Delta)$ is positive semidefinite. This concludes the proof because $\Delta \succeq 0$ implies that \eqref{eq:LMI2} is satisfied. 
\end{proof}
The optimization problem in \eqref{eq:convex SLS} is amenable to standard solvers because, as we show in the next Proposition, it is a convex program. The proof is deferred to Appendix \ref{app:proof_convexity}.
\begin{proposition}
\label{prop:convexity_problem}
    The optimization problem in \eqref{eq:convex SLS} has a linear objective. Moreover, the constraints are convex sets in the optimization variables $(\lambda,\ Q,\ q,\ c,\ \mathbf{\Phi},\ \bm\phi,\ s_i,\ z_i,\ \zeta_i,\ \tau,\ \sigma),\ \forall i\in[N]$. Hence, \eqref{eq:convex SLS} is convex. 
\end{proposition}
We conclude the section noticing that \eqref{eq:convex SLS}, due to the nonlinear constraint \eqref{eq:logdet inequality}, is a conic problem when $\epsilon\neq0$. This is in contrast with the Wasserstein counterpart of \eqref{eq:convex SLS}, which is a semidefinite program \cite{kuhn2019wasserstein}.
\section{Numerical Results}
\label{sec:numerical}

\subsection{Trajectory Planning Boeing B-747}
In this section, we benchmark the policy $\bm \pi_S$ synthesized using \eqref{thm:main} against the Wasserstein counterpart $\bm \pi_W$, and a certainty-equivalence policy $\bm \pi_\text{emp}$, which is an instance of $\bm \pi_S$ when $\rho, \epsilon\rightarrow 0$ in \eqref{eq:convex SLS}. For our experiments,\footnote{All our experiments were run on an M3 Pro CPU machine with 36GB RAM. All SDP problems were modeled in Matlab 2025b using Yalmip and solved with MOSEK.
Our source code is publicly available at \href{https://github.com/DecodEPFL/Sinkhorn_DRCC}{\texttt{https://github.com/DecodEPFL/Sinkhorn\_DRCC.git}}.} we consider the two-input fourth order plant obtained by discretizing with sampling period $T_s = \SI{0.1}{\s}$ a continuous-time model of a Boeing B-747 aircraft flying at an altitude of $\SI{20000}{ft}$ with a speed of Mach 0.8 \cite{ishihara1992design}. The state and input vector are defined as $x(t) = [\beta(t)\ p(t)\ r(t)\ \phi(t)]^\top$ and $u(t) = [\delta_a(t)\ \delta_r(t)]^\top$, respectively, where $\beta(t)$ is the sideslip angle, $p(t)$ is the roll rate, $r(t)$ is the yaw rate, $\phi(t)$ is the roll angle, $\delta_a(t)$ is the aileron deflection and $\delta_r(t)$ is the rudder deflection. We measure all angles in radians and all angular velocities in radians per second. The state-space model matrices describing the discrete-time system dynamics are given by
\begin{equation*}
    \begin{aligned}
A &= \begin{bmatrix} 
0.9801 & 0.0003 & -0.0980 & 0.0038 \\ 
-0.3868 & 0.9071 & 0.0471 & -0.0008 \\ 
0.1591 & -0.0015 & 0.9691 & 0.0003 \\ 
-0.0198 & 0.0958 & 0.0021 & 1.000 
\end{bmatrix}, \\
B &= \begin{bmatrix} 
-0.0001 & 0.0058 \\ 
0.0296 & 0.0153 \\ 
0.0012 & -0.0908 \\ 
0.0015 & 0.0008 
\end{bmatrix}\,.
\end{aligned}
\end{equation*}
We fix $T=10$ and $D = \text{diag}(I_4, 0.01)$.
We consider the goal of steering the system from a given initial condition $x_0 = \begin{bmatrix}
    0.7\ 0.1\ 0.5\ 0.3
\end{bmatrix}^\top$ to the terminal constraint set defined by the polytope $\mathcal X = \{ x \in \mathbb{R}^4 \mid |x| \leq x_{\text{max}} \}$, where $x_{\text{max}} = \begin{bmatrix}
0.3491\ 0.2618\ 0.1745\  0.5236
\end{bmatrix}^\top$. Moreover, we set the constraints violation probability in \eqref{eq:Sinkhorn_problem} to $\gamma = 0.3$. For simplicity, we do not consider input constraints, that is, we let $\mathcal U = \R^2$.

Airplanes are heavily affected by wind turbulence, which acts as an additive noise to the system. Most commonly, the atmospheric turbulence is represented as the convolution of white noise through an
LTI shaping filter, usually referred to as a Dryden model \cite{milhdbk1797_short, gage2003creating}. The continuous filter for the lateral velocity, roll and yaw angular velocities are given by
\begin{align*}
    H_v(s) &= \sigma_v \sqrt{\frac{2L_v}{\pi V}} \frac{1 + \frac{2\sqrt{3}L_v}{V}s}{\left( 1 + \frac{2L_v}{V}s \right)^2}\,,\\
    H_p(s) &= \sigma_w \sqrt{\frac{0.8}{V}} \frac{\left( \frac{\pi}{4b} \right)^{1/6}}{(2L_w)^{1/3} \left( 1 + \frac{4b}{\pi V} s \right)}\,,\\
    H_r(s) &= \frac{- \frac{s}{V}}{\left( 1 + \frac{3b}{\pi V}s \right)} \cdot H_v(s)\,, 
\end{align*}
respectively, where $V$ is the aircraft velocity, $\sigma_v = \sigma_w = \SI{20}{ft\per\s}$ are turbulence intensities, $L_v = L_w = \SI{875}{ft}$ are turbulence
scale lengths, and $b = \SI{210}{ft}$ indicates wing span. 

To design the different controllers, we generate $n=5$ trajectories of length $T$ of the disturbance. Such samples are obtained simulating the discretization of the previous filters in response to white noise. In particular, $H_p$ and $H_r$ yield the noise realizations along the $p$ and $r$ state components, respectively; the noise on $\beta$ is approximately the ratio of the output of $H_v$ and the velocity of the plane $V$; finally the noise on $\phi$ is obtained integrating the output of $H_p$. We remark that such noise profile is colored and hence violates standard white noise assumptions used in stochastic control. We first compute the policies $\bm \pi_S$ and $\bm \pi_W$ corresponding to different radii $\rho\in\{0.001,\, 0.003,\,0.007\}$ and regularization parameters $\epsilon\in\{2,\, 2.5,\, 3.2,\, 4,\, 5,\, 6.3,\, 8,\, 10,\, 12.5,\, 16,\, 20\}\times 10^{-6}$.

Then, for the testing phase, we generate $20000$ new unseen noise trajectories of the wind using the same Dryden model described previously. 
\cref{fig:comparison} compares the cost and constraint violations obtained by the different controllers in response to such testing samples. We observe that both DR approaches outperform the nominal design $\bm \pi_{\text{emp}}$, which pays substantially higher control cost and incurs a large constraint violation rate when deployed on unseen test data. Moreover, the benefit of introducing an entropic regularization in the design of $\bm\pi_S$ can be seen in the reduction of conservatism; indeed, while $\bm\pi_S$ incurs a slightly higher constraint violation rate than $\bm\pi_W$ for $\rho = 0.001$, $\bm\pi_S$ consistently meets the required safety threshold $1- \gamma$, and does so with a substantially lower control cost.

This can also be observed in \cref{fig:trajectories}, where terminal states are projected onto the angle and angular velocity components. We see that $\bm \pi_{\text{emp}}$ showcases a large violation rate, failing to comply with the design requirements. In contrast, both $\bm\pi_S$ and $\bm\pi_W$ guarantee that all system trajectories remain inside the feasible region. In addition, since the $\mathcal{S}$-set is smaller according to \cref{proposition: relationships}, the level of conservatism is reduced and $\bm\pi_S$ can effectively stay closer to the border of the gray feasible area.
\begin{figure}[h]
    \centering
    \includegraphics[trim=0 0 0 0, clip, width = \columnwidth]{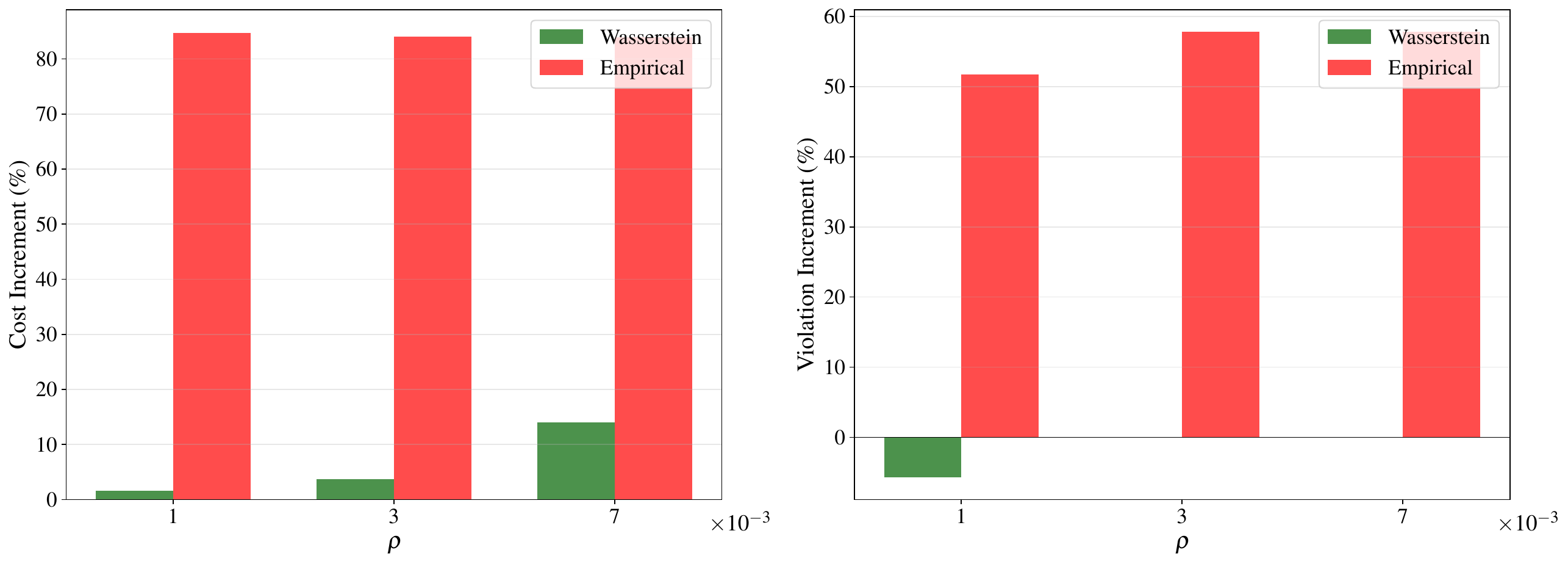}
    \caption{Comparison in cost (left) and violation (right) percentage increase of the Wasserstein DRC in green and empirical controller in red with respect to the Sinkhorn DRC. The depicted Sinkhorn baseline represents the best performance across all tested values of $\epsilon$ in response to unseen wind turbulence realizations.}
    \label{fig:comparison}
\end{figure}
\begin{figure}[h]
    \centering
    \includegraphics[trim=0 0 0 0, clip, width=\columnwidth]{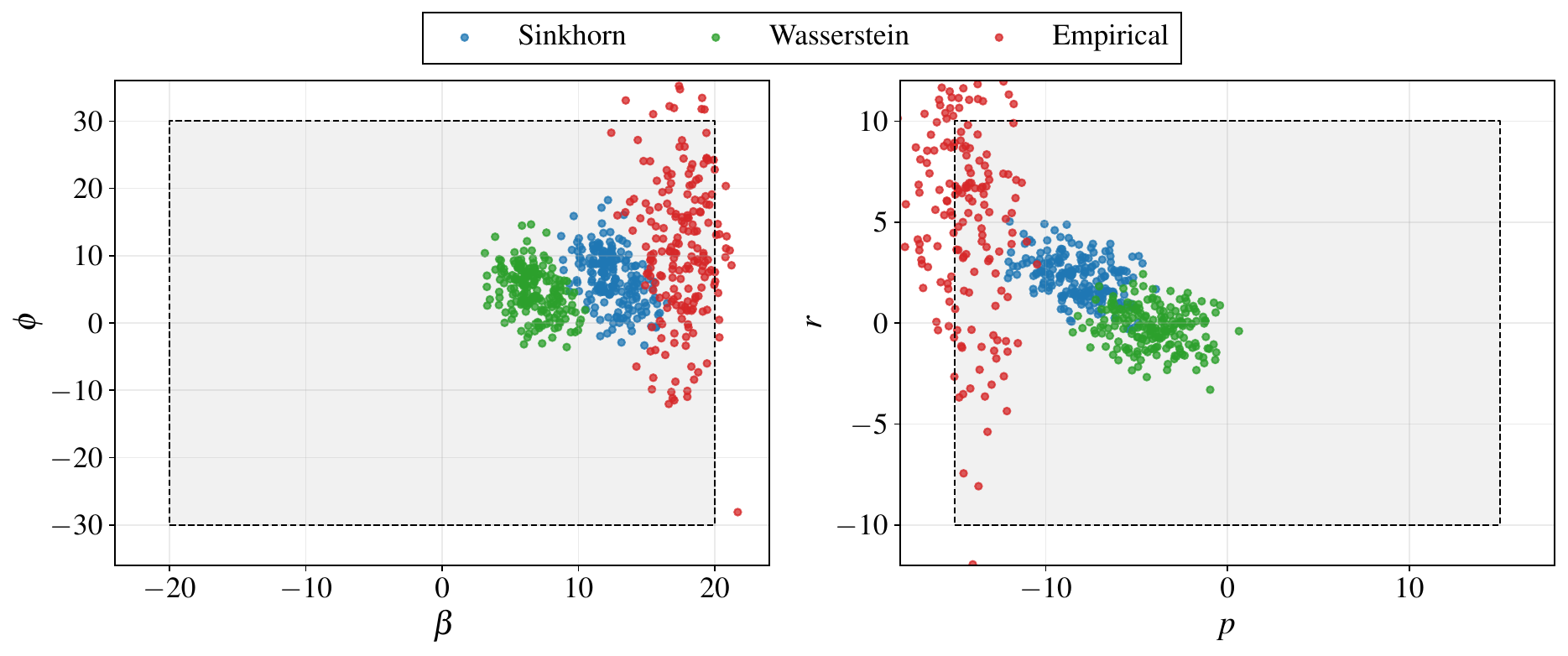}
    \caption{Terminal state of different trajectories given by the policies $\bm \pi_S$ (blue), $\bm\pi_W$ (green), and $\bm\pi_{\text{emp}}$ (red) in response to samples taken from a  random probability in the Sinkhorn ambiguity set. The gray-shaded area represents the terminal feasible set.}
\label{fig:trajectories}
\end{figure}
\subsection{Implementation and complexity analysis}
In \cref{subsection:reformulation}, we established that problem \eqref{eq:convex SLS} is jointly convex in the optimization variables. However, standard modeling toolboxes (e.g., YALMIP) are unable to directly represent this formulation. Since the problem is scalar in the dual variable $\lambda$, we address this limitation by performing a golden section search over the interval $[0,\lambda_{\max}]$, where $\lambda_{\max}$ is a sufficiently large upper bound on the optimal dual variable. At each iteration, the objective function is evaluated by fixing $\lambda$ and solving the resulting conic optimization problem over the remaining variables. The overall computational complexity of solving \eqref{eq:convex SLS} can therefore be decomposed into two parts. The first corresponds to the golden section search, which is known to require $\mathcal{O}(\log(1/\upsilon))$ iterations to achieve an accuracy $\upsilon>0$. The second corresponds to the cost of evaluating the objective function at each iteration, which involves solving a conic optimization problem with $\mathcal{O}(T^2(p^2 + pm + pd) + N)$ decision variables.

\section{Conclusion}
In this paper, we studied topological and variational properties of Sinkhorn ambiguity sets, proving their convexity and weak-compactness. We illustrated the relevance of these results in a finite-horizon Sinkhorn DR control problem with DR safety constraints. Leveraging convexity and weak compactness of Sinkhorn ambiguity sets, we characterized the DR safety constraints as convex sets, and we presented an equivalent, tractable reformulation of the resulting DR controller synthesis problem over affine feedback policies using the system level parametrization. Our numerical experiments  confirm that introducing an entropic regularization in the optimal transport definition allows reducing conservatism, especially when only a limited number of noise samples is available. Promising directions for future work include the development of gradient-based methods to efficiently solve the Sinkhorn DR control problem and extensions of the proposed approach to nonlinear systems.

\bibliographystyle{IEEEtran}
\bibliography{reference}

\begin{thebibliography}{10}
\providecommand{\url}[1]{#1}
\csname url@samestyle\endcsname
\providecommand{\newblock}{\relax}
\providecommand{\bibinfo}[2]{#2}
\providecommand{\BIBentrySTDinterwordspacing}{\spaceskip=0pt\relax}
\providecommand{\BIBentryALTinterwordstretchfactor}{4}
\providecommand{\BIBentryALTinterwordspacing}{\spaceskip=\fontdimen2\font plus
\BIBentryALTinterwordstretchfactor\fontdimen3\font minus \fontdimen4\font\relax}
\providecommand{\BIBforeignlanguage}[2]{{%
\expandafter\ifx\csname l@#1\endcsname\relax
\typeout{** WARNING: IEEEtran.bst: No hyphenation pattern has been}%
\typeout{** loaded for the language `#1'. Using the pattern for}%
\typeout{** the default language instead.}%
\else
\language=\csname l@#1\endcsname
\fi
#2}}
\providecommand{\BIBdecl}{\relax}
\BIBdecl

\bibitem{bertsekas2012dynamic}
D.~Bertsekas, \emph{Dynamic programming and optimal control: Volume I}.\hskip 1em plus 0.5em minus 0.4em\relax Athena scientific, 2012, vol.~4.

\bibitem{zhou1998essentials}
K.~Zhou and J.~Doyle, \emph{Essentials of robust control}.\hskip 1em plus 0.5em minus 0.4em\relax Upper Saddle River: Prentice-Hall, 1998.

\bibitem{bernstein1988lqg}
D.~S. Bernstein and W.~M. Haddad, ``{LQG} control with an $\mathcal{H}_\infty$ performance bound: A {R}iccati equation approach,'' in \emph{1988 American Control Conference}.\hskip 1em plus 0.5em minus 0.4em\relax IEEE, 1988, pp. 796--802.

\bibitem{doyle1989optimal}
J.~Doyle, K.~Zhou, and B.~Bodenheimer, ``Optimal control with mixed $\mathcal{H}_2$ and $\mathcal{H}_\infty$ performance objectives,'' in \emph{1989 American Control Conference}.\hskip 1em plus 0.5em minus 0.4em\relax IEEE, 1989, pp. 2065--2070.

\bibitem{martin2025guarantees}
A.~Martin, L.~Furieri, F.~D{\"o}rfler, J.~Lygeros, and G.~Ferrari-Trecate, ``On the guarantees of minimizing regret in receding horizon,'' \emph{IEEE Transactions on Automatic Control}, vol.~70, no.~3, pp. 1547--1562, 2025.

\bibitem{goel2023regret}
G.~Goel and B.~Hassibi, ``Regret-optimal estimation and control,'' \emph{IEEE Transactions on Automatic Control}, vol.~68, no.~5, pp. 3041--3053, 2023.

\bibitem{martin2024regret}
A.~Martin, L.~Furieri, F.~D{\"o}rfler, J.~Lygeros, and G.~Ferrari-Trecate, ``Regret optimal control for uncertain stochastic systems,'' \emph{European Journal of Control}, vol.~80, p. 101051, 2024.

\bibitem{van2015distributionally}
B.~P. Van~Parys, D.~Kuhn, P.~J. Goulart, and M.~Morari, ``Distributionally robust control of constrained stochastic systems,'' \emph{IEEE Transactions on Automatic Control}, vol.~61, no.~2, pp. 430--442, 2015.

\bibitem{falconi2025distributionally}
L.~Falconi, A.~Ferrante, and M.~Zorzi, ``Distributionally robust {LQG} control under distributed uncertainty,'' \emph{Automatica}, vol. 174, p. 112128, 2025.

\bibitem{petersen2000minimax}
I.~R. Petersen, M.~R. James, and P.~Dupuis, ``Minimax optimal control of stochastic uncertain systems with relative entropy constraints,'' \emph{IEEE Transactions on Automatic Control}, vol.~45, no.~3, pp. 398--412, 2000.

\bibitem{fochesato2025distributionally}
M.~Fochesato, L.~Falconi, M.~Zorzi, A.~Ferrante, and J.~Lygeros, ``Distributionally robust {LQG} with {K}ullback-{L}eibler ambiguity sets,'' \emph{arXiv preprint arXiv:2505.08370}, 2025.

\bibitem{mohajerin2018data}
P.~Mohajerin~Esfahani and D.~Kuhn, ``Data-driven distributionally robust optimization using the {W}asserstein metric: Performance guarantees and tractable reformulations,'' \emph{Mathematical Programming}, vol. 171, no.~1, pp. 115--166, 2018.

\bibitem{taskesen2024distributionally}
B.~Taskesen, D.~Iancu, {\c{C}}.~Ko{\c{c}}yi{\u{g}}it, and D.~Kuhn, ``Distributionally robust linear quadratic control,'' \emph{Advances in Neural Information Processing Systems}, vol.~36, 2024.

\bibitem{lanzetti2024optimality}
N.~Lanzetti, A.~Terpin, and F.~D{\"o}rfler, ``Optimality of linear policies for distributionally robust linear quadratic gaussian regulator with stationary distributions,'' \emph{arXiv preprint arXiv:2410.22826}, 2024.

\bibitem{aolaritei2023wasserstein}
L.~Aolaritei, M.~Fochesato, J.~Lygeros, and F.~D{\"o}rfler, ``Wasserstein tube {MPC} with exact uncertainty propagation,'' in \emph{2023 62nd IEEE Conference on Decision and Control (CDC)}.\hskip 1em plus 0.5em minus 0.4em\relax IEEE, 2023, pp. 2036--2041.

\bibitem{coulson2021distributionally}
J.~Coulson, J.~Lygeros, and F.~D{\"o}rfler, ``Distributionally robust chance constrained data-enabled predictive control,'' \emph{IEEE Transactions on Automatic Control}, vol.~67, no.~7, pp. 3289--3304, 2021.

\bibitem{micheli2022data}
F.~Micheli, T.~Summers, and J.~Lygeros, ``Data-driven distributionally robust {MPC} for systems with uncertain dynamics,'' in \emph{2022 IEEE 61st Conference on Decision and Control (CDC)}.\hskip 1em plus 0.5em minus 0.4em\relax IEEE, 2022, pp. 4788--4793.

\bibitem{brouillon2025distributionally}
J.-S. Brouillon, A.~Martin, J.~Lygeros, F.~D{\"o}rfler, and G.~Ferrari-Trecate, ``Distributionally robust infinite-horizon control: from a pool of samples to the design of dependable controllers,'' \emph{IEEE Transactions on Automatic Control}, vol.~70, no.~10, pp. 6465--6480, 2025.

\bibitem{gao2023distributionally}
R.~Gao and A.~Kleywegt, ``Distributionally robust stochastic optimization with {W}asserstein distance,'' \emph{Mathematics of Operations Research}, vol.~48, no.~2, pp. 603--655, 2023.

\bibitem{cuturi2013sinkhorn}
M.~Cuturi, ``Sinkhorn distances: Lightspeed computation of optimal transport,'' \emph{Advances in neural information processing systems}, vol.~26, 2013.

\bibitem{azizian2023regularization}
W.~Azizian, F.~Iutzeler, and J.~Malick, ``Regularization for {W}asserstein distributionally robust optimization,'' \emph{ESAIM: Control, Optimisation and Calculus of Variations}, vol.~29, p.~33, 2023.

\bibitem{blanchet2023unifying}
J.~Blanchet, D.~Kuhn, J.~Li, and B.~Taskesen, ``Unifying distributionally robust optimization via optimal transport theory,'' \emph{arXiv preprint arXiv:2308.05414}, 2023.

\bibitem{dapogny2023entropy}
C.~Dapogny, F.~Iutzeler, A.~Meda, and B.~Thibert, ``Entropy-regularized {W}asserstein distributionally robust shape and topology optimization,'' \emph{Structural and Multidisciplinary Optimization}, vol.~66, no.~3, p.~42, 2023.

\bibitem{cescon2025data}
R.~Cescon, A.~Martin, and G.~Ferrari-Trecate, ``Data-driven {D}istributionally {R}obust {C}ontrol {B}ased on {S}inkhorn {A}mbiguity {S}ets,'' in \emph{2025 IEEE 64th Conference on Decision and Control (CDC)}.\hskip 1em plus 0.5em minus 0.4em\relax IEEE, 2025, pp. 4708--4713.

\bibitem{feng2026sinkhorn}
Y.~Feng, X.~Li, S.~X. Ding, H.~Ye, and C.~Shang, ``Sinkhorn {D}istributionally {R}obust {S}tate {E}stimation via {S}ystem {L}evel {S}ynthesis,'' \emph{arXiv preprint arXiv:2602.08018}, 2026.

\bibitem{kuhn2025distributionally}
D.~Kuhn, S.~Shafiee, and W.~Wiesemann, ``Distributionally robust optimization,'' \emph{Acta Numerica}, vol.~34, pp. 579--804, 2025.

\bibitem{sinkhorn}
J.~Wang, R.~Gao, and Y.~Xie, ``Sinkhorn distributionally robust optimization,'' \emph{Operations Research}, 2025.

\bibitem{villani2009optimal}
C.~Villani \emph{et~al.}, \emph{Optimal transport: old and new}.\hskip 1em plus 0.5em minus 0.4em\relax Springer, 2009, vol. 338.

\bibitem{kallenberg2002foundations}
O.~Kallenberg, \emph{Foundations of modern probability}, 2nd~ed.\hskip 1em plus 0.5em minus 0.4em\relax Springer-Verlag, New York, 2002.

\bibitem{kuhn2019wasserstein}
D.~Kuhn, P.~M. Esfahani, V.~A. Nguyen, and S.~Shafieezadeh-Abadeh, ``Wasserstein distributionally robust optimization: Theory and applications in machine learning,'' in \emph{Operations research \& management science in the age of analytics}.\hskip 1em plus 0.5em minus 0.4em\relax Informs, 2019, pp. 130--166.

\bibitem{shafieezadeh2023new}
S.~Shafieezadeh-Abadeh, L.~Aolaritei, F.~D{\"o}rfler, and D.~Kuhn, ``New perspectives on regularization and computation in optimal transport-based distributionally robust optimization,'' \emph{arXiv preprint arXiv:2303.03900}, 2023.

\bibitem{yue}
M.-C. Yue, D.~Kuhn, and W.~Wiesemann, ``On linear optimization over {W}asserstein balls,'' \emph{Mathematical Programming}, vol. 195, no.~1, pp. 1107--1122, 2022.

\bibitem{rockafellar2000optimization}
R.~T. Rockafellar and S.~Uryasev, ``Optimization of conditional value-at-risk,'' \emph{Journal of risk}, vol.~2, pp. 21--42, 2000.

\bibitem{cescon2025global}
R.~Cescon, A.~Martin, and G.~Ferrari-Trecate, ``On the global optimality of linear policies for {S}inkhorn distributionally robust linear quadratic control,'' \emph{arXiv preprint arXiv:2509.00956}, 2025.

\bibitem{ishihara1992design}
T.~Ishihara, H.-J. Guo, and H.~Takeda, ``A design of discrete-time integral controllers with computation delays via loop transfer recovery,'' \emph{Automatica}, vol.~28, no.~3, pp. 599--603, 1992.

\bibitem{milhdbk1797_short}
\emph{{U.S. Military Handbook}}, U.S. Department of Defense, 1997.

\bibitem{gage2003creating}
S.~Gage, ``Creating a unified graphical wind turbulence model from multiple specifications,'' in \emph{AIAA Modeling and simulation technologies conference and exhibit}, 2003, p. 5529.

\bibitem{rockafellar2009variational}
R.~T. Rockafellar and R.~J.-B. Wets, \emph{Variational analysis}.\hskip 1em plus 0.5em minus 0.4em\relax Springer Science \& Business Media, 2009, vol. 317.

\bibitem{billingsley}
P.~Billingsley, \emph{Convergence of probability measures}.\hskip 1em plus 0.5em minus 0.4em\relax John Wiley \& Sons, 2013.

\bibitem{munkres}
J.~Munkres, \emph{Topology}, ser. Featured {Titles} for {Topology}.\hskip 1em plus 0.5em minus 0.4em\relax Prentice Hall, Incorporated, 2000.

\bibitem{mallasto2022entropy}
A.~Mallasto, A.~Gerolin, and H.~Q. Minh, ``Entropy-regularized 2-{W}asserstein distance between {G}aussian measures,'' \emph{Information Geometry}, vol.~5, no.~1, pp. 289--323, 2022.

\bibitem{polyanskiy2014}
Y.~Polyanskiy and Y.~Wu, ``Lecture notes on information theory,'' \emph{Lecture Notes for ECE563 (UIUC) and}, vol.~6, no. 2012-2016, p.~7, 2014.

\bibitem{clement2008wasserstein}
P.~Cl{\'e}ment and W.~Desch, ``Wasserstein metric and subordination,'' \emph{Studia Mathematica}, vol.~1, no. 189, pp. 35--52, 2008.

\bibitem{wang2019system}
Y.-S. Wang, N.~Matni, and J.~C. Doyle, ``A system-level approach to controller synthesis,'' \emph{IEEE Transactions on Automatic Control}, vol.~64, no.~10, pp. 4079--4093, 2019.

\bibitem{Sion1958}
M.~Sion, ``On general minimax theorems,'' \emph{Pacific Journal of Mathematics}, vol.~8, no.~1, pp. 171--176, 1958.

\bibitem{boyd2004convex}
S.~Boyd and L.~Vandenberghe, \emph{Convex optimization}.\hskip 1em plus 0.5em minus 0.4em\relax Cambridge university press, 2004.

\end{thebibliography}

\appendices
\section{Auxiliary Results}
\label{appendix}
We review some well-known facts from measure theory and variational analysis that are used to prove our results. Later in the section we also present a brief overview of the topology of the Wasserstein space. 
\begin{definition}[Lower limit and lower semicontinuity]
The lower limit of a function $f: \R^n \rightarrow \overline{\R} = \R\cup\{\infty\}$ at $\bar{x}$ is the value in $\overline{\R}$ defined by
\begin{align*}
\liminf_{x\rightarrow\bar{x}} f(x) := \lim_{\delta\rightarrow 0}\,\left[ \inf_{x\in B(\bar{x}, \delta)}\, f(x)\right]
=\sup_{\delta>0}\, \left[ \inf_{x\in B(\bar{x}, \delta)}\, f(x)\right]
\end{align*}
where $B(\bar{x}, \delta) := \{x\ |\ \|x- \bar{x}\|\leq \delta)\}$.
\\
A function $f: \R^n \rightarrow \overline{\R}$ is lower semicontinuous (lsc) at $\bar{x}$ if \looseness=-1
\begin{equation*}
    \liminf_{x\rightarrow\bar{x}} f(x) \geq f(\bar{x})\,,
\end{equation*}
and lower semicontinuous on $\R^n$ if it holds for every $\bar{x}\in\R^n$.
\end{definition}

We next introduce the definition of epigraph which can be understood as the collection of all the points of $\R^{n+1}$ lying on or above the graph of $f$. Complementary to this is the sublevel set, which identifies the region of the domain where the function does not exceed a threshold $\alpha\in\R$.
\begin{definition}[Epigraph and sublevel set]
For $f: \R^n \rightarrow \overline{\R}$, the epigraph of $f$ is the set 
\begin{equation*}
\epi f = \{(x,\alpha) \in \R^n \times \R\ |\ f(x)\leq \alpha\}\,,
\end{equation*}
and the sublevel set of level $\alpha\in\R$ is 
\begin{equation*}
    \text{lev}_{\leq\alpha} := \{x\in\R^n\ |\ f(x) \leq \alpha\}\,.
\end{equation*}
\end{definition}

The following Theorem characterizes the closedness of the level sets of $f$ in terms of properties of the function itself.
\begin{theorem}[{\cite[Theorem 1.6]{rockafellar2009variational}}] The following properties of a function $f: \R^n \rightarrow \overline{\R}$ are equivalent:
\begin{enumerate}[label=(\roman*)]
    \item $f$ is lower semicontinuous on $\R^n$;
    \item the epigraph set $\epi f$ is closed in $\R^n\times\R$;
    \item the level sets of type $\text{lev}_{\leq \alpha}f$ are all closed in $\R^n$.
\end{enumerate}
\label{thm:rockafellar}
\end{theorem}

We now give some useful definitions and topological properties regarding probability distributions. In the following we assume that $\mathcal{Z}$ is a
closed subset of $\R^n$, hence we know that $\Z$ is a Polish space, i.e. a separable
and completely metrizable topological space.
\begin{definition}[Weak convergence of probability distributions] A sequence of probability distributions $\{\Prob_j\}_{j\in\N}\in\Pset$, converges 
weakly to $\Prob\in\Pset$ if for any bounded and continuous function $\phi$ on $\Z$, we have
\begin{equation*}
    \lim_{j\rightarrow\infty}\int_\Z \phi(z)d\Prob_j(z) = \int_\Z \phi(z)d\Prob(z)\,.
\end{equation*}
\end{definition}
To prove some results in the paper we equip the space of probability distributions $\Pset$ with the weak topology, i.e. the topology generated by the
open sets given by
\begin{equation*}
    \mathcal{U}_{f,\delta} = \{\Prob\in\Pset\ :\ |\E_\Prob[\phi(Z)]| < \delta\}\,,
\end{equation*}
for any continuous and bounded function $\phi$ on $\Z$ and tolerance $\delta>0$.

We will also use the concept of tightness regarding a family of probability distributions which is widely used in the theory of weak convergence 
and in its applications.
\begin{definition}[Tightness]
A family of probability distributions $\mathcal{S}\subseteq\Pset$ is tight if for every $\epsilon > 0$ there exists a compact subset $K\subseteq\Z$ 
such that $\Prob(Z\notin K) \leq \epsilon$ for all $\Prob\in\mathcal{S}$.
\end{definition}
Intuitively, a family of probability distributions $\mathcal{S}$ is tight if, for any $\epsilon > 0$, there exists a compact set $K$ such that every distribution in the family assigns at least $1 - \epsilon$ of its mass to $K$. This requirement prevents the probability mass from ``escaping'' to infinity or vanishing in the limit.
\begin{example}
For $\mathcal{Z} = \R$ the family of uniform distributions over $[-\frac{1}{n}, \frac{1}{n}]\,, n\in\N$ is tight because the compact set $[-1, 1]$ contains all the probability mass. On the contrary the family $\mathcal{S}$ of uniform distributions over $[-n, n]\,, n\in\N$ is not tight. Indeed, given a tolerance $\epsilon=0.5$ and any compact set $K$ (which is necessarily contained in some interval $[-M, M]$), any distribution in $\mathcal{S}$ with $n>2M$ assigns more than $\epsilon$ mass outside of $K$.
\end{example}
Next, we introduce two notions of compactness of a probability space under the weak topology.
\begin{definition}[Weak compactness]
\label{def:compactness}
We say that a family of probability distributions $\mathcal{S}\subseteq\Pset$ is weakly 
compact if every collection of open subsets of $\mathcal{S}$ whose union is $\mathcal{S}$ contains a finite subcollection whose union is also $\mathcal{S}$.
\end{definition}
\begin{definition}[Weak sequential compactness] We say that a family of probability distributions $\mathcal{S}\subseteq\Pset$ is weakly sequentially 
    compact if every sequence in $\mathcal{S}$ has a subsequence converging to an element of $\mathcal{S}$.
\end{definition}
The concepts of tightness and weak sequential compactness are strictly related by Prokhorov's Thereom.
\begin{theorem}[Prokhorov's theorem, {\cite[Theorem 5.1]{billingsley}}]
A collection of probability distributions $\mathcal{S}\subseteq\Pset$ is tight if and only if the weak closure of $\mathcal{S}$ is weakly sequentially 
compact in $\Pset$.
\end{theorem}
Note that the space $\Pset$ is metrizable, hence the notions of compactness and sequential compactness are equivalent \cite[Theorem 28.2]{munkres}.
\\
Now, we provide a topological characterization of Wasserstein ambiguity sets which was first presented in \cite[Theorem 1]{yue}.
\begin{theorem}[Weak compactness of Wasserstein balls] The Wasserstein ambiguity set $\mathbb{W}_\rho(\Prob)$ is weakly compact whenever the reference
measure has finite $p$-th moment.
\label{thm:weak_compactness_Wasserstein}
\end{theorem}
\section{Proofs for Section \ref{sec:distributional_uncertainty}}
\subsection{Proof of Proposition \ref{proposition:convexity}}
\label{appendix: convexity}
We first prove the convexity of the Sinkhorn discrepancy defining the $\mathcal{S}$-set in \eqref{eq:ambiguity_set}. We notice that \eqref{eq:sinkhorn} is equivalent to
\begin{equation*}
    \underbrace{%
    \inf_{\gamma \in \Gamma(\Prob, \Q)} \left\{ \mathbb{E}_\gamma[c(x, y)] + \epsilon\, \mathrm{KL}(\gamma \| \Prob \times \Q) \right\}%
    }_{(\heartsuit)}
    + \mathrm{KL}(\Q \| \nu)\,.
\end{equation*}
Hence, $W_c^\epsilon(\Prob, \Q)$ is convex in $\Q$ because sum of convex functions. Indeed, \cite[Proposition 1]{mallasto2022entropy} proved the convexity of $(\heartsuit)$, while the $\mathrm{KL}$ divergence is convex in its first argument. 

Now, consider two probability distributions in the $\mathcal{S}$-set, $\Q_1, \Q_2$ and a scalar parameter $\alpha\in(0,1)$. We show that $\alpha \Q_1 + (1-\alpha)\Q_2\in\B_{\rho, \epsilon}(\Prob)$. 
\begin{align*}
&W_c^{\epsilon}(\Prob, \alpha\Q_1+(1-\alpha)\Q_2) 
\\
&\leq \alpha W_c^{\epsilon}(\Prob, \Q_1)+(1-\alpha)W_c^{\epsilon}(\Prob, \Q_2)
\\
&\leq \alpha\rho + (1-\alpha)\rho = \rho\,,
\end{align*}
where the first inequality follows from the convexity of the Sinkhorn discrepancy.
\subsection{Proof of Lemma \ref{lemma:existance}}
\label{appendix: solvability}
By \cite[Corollary 3.16]{kuhn2025distributionally}, the set $\Gamma(\Prob, \Q)$ is weakly compact. In addition, the transportation cost function $c(x, y)$ is lower 
semicontinuous and bounded from below. Therefore, by \cite[Proposition 3.3]{kuhn2025distributionally}, the expected value $\E_\gamma[c(X, Y)]$ is weakly lower semicontinuous.
Moreover, by \cite[Theorem 3.6]{polyanskiy2014}, the KL divergence term in (\ref{eq:sinkhorn}) is lower semicontinous in $\gamma$ and has lower bound 
zero. Since the sum of two lower semicontinuous functions that are bounded from below remains lower semicontinuous and bounded from below, the Sinkhorn
optimal transport problem in (\ref{eq:sinkhorn}) is solvable thanks to Weierstrass' theorem, and its infimum is attained.
\subsection{Proof of Lemma \ref{lemma: semicontinuity}}
\label{appendix: semicontinuity}
Recall that, to prove the lemma, we aim to show that
\begin{equation}
\label{eq:liminf}
    \liminf_{j\rightarrow\infty}W_c^{\epsilon}(\Prob_j, \Q_j) 
    \geq W_c^{\epsilon}(\Prob, \Q)\,,
\end{equation}
for any two sequences of probability distributions $\Prob_j$ and $\Q_j$, indexed by $j\in\N$, that converge weakly to $\Prob$ and $\Q$, respectively. To do so, we (i) consider the optimal coupling $
\gamma_j^\star$ between $\Prob_j$ and $\Q_j$ and show that $\lim_j \gamma_j^\star$ yields a feasible but possibly suboptimal coupling $\gamma$ between $\Prob$ and $\Q$; and later (ii) we use such $\gamma$ to construct a lower bound to the limit inferior in \eqref{eq:liminf} which serves also as an upper bound to $W_c^\epsilon(\Prob, \Q)$, from which the lemma follows.
\newline
The proof of (i) is inspired by \cite[Theorem 1]{yue} and \cite[Theorem 5.2]{clement2008wasserstein}. We need to construct a converging sequence of optimal couplings. To do so, we first show that the union over $j\in\N$ of all joint couplings between $\Prob_j$ and $\Q_j$ is weakly compact. Indeed, consider the sets $\mathcal{P} = \{\Prob_j\}_{j\in\N}$ and $\mathcal{Q} = \{\Q_j\}_{j\in\N}$. Since we have assumed weak convergence of the previous sequences, each subsequence converges. Hence, the weak closures of such sets are weakly sequentially compact which in turns implies weak compactness. Prokhorov's theorem (see Appendix \ref{appendix}) implies that $\mathcal{P}, \mathcal{Q}$ are both tight. Hence, given any $\epsilon > 0$ there exist two compact sets $\mathcal{C}, \hat{\mathcal{C}} \subseteq \R^d$ such that $\forall j\in\N$
\begin{equation*}
    \Prob_j(Z\notin \mathcal{C}) \leq \epsilon/2 \quad \text{and}\quad \Q_j(Z\notin \hat{\mathcal{C}}) \leq \epsilon/2\,.
\end{equation*}
Next, we can show that
\begin{equation}
\label{eq:union}
    \bigcup_{j\in\N} \Gamma(\Prob_j, \Q_j)
\end{equation}
is also tight. Indeed, taken $\mathcal{C} \times \hat{C}$ compact and $\epsilon$ arbitrary, whenever $\gamma\in\Gamma(\Prob_j, \Q_j)$ for some $j\in\N$, we have
\begin{equation*}
    \gamma((Z, \hat{Z}) \notin \mathcal{C} \times \hat{C}) \leq  \Prob_j(Z\notin \mathcal{C}) +\Q_j(Z\notin \hat{\mathcal{C}}) \leq \epsilon\,. 
\end{equation*}
Again by Prokhorov's theorem the family of probability distributions in (\ref{eq:union}) is relatively compact, i.e. its closure is compact. 
\newline
Consider $\gamma_j^{\star}$ to be 
an optimal coupling of $\Prob_j$ and $\Q_j$ solving (\ref{eq:sinkhorn}) which we know it exists because of Lemma \ref{lemma:existance}. Since the ensemble of such
optimal couplings belongs to a weakly compact set, i.e. the weak closure of (\ref{eq:union}), we may assume without loss of generality that the sequence 
$\gamma^{\star}_j$, $j\in\N$, converges weakly to some distribution $\gamma$, otherwise we can pass to a subsequence which is guaranteed to converge 
by compactness. Then, $\gamma\in\Gamma(\Prob, \Q)$. Indeed, it can be shown that $\gamma$ has marginals $\Prob$ and $\Q$. Consider any given bounded and 
continuous function $g:\mathcal{Z}\rightarrow\R$, then
\setcounter{relctr}{0}
\begin{align*}
&\int_\mathcal{Z}g(z)d\Prob(z) 
\labelrel={marginal:first_eq}
\lim_{j\rightarrow\infty}\int_\mathcal{Z}g(z)d\Prob_j(z)
\\
&\labelrel={marginal:second_eq}
\lim_{j\rightarrow\infty}\int_\mathcal{Z\times\mathcal{Z}}g(z)d\gamma^{\star}_j(z, \hat{z})
\labelrel={marginal:third_eq}
\int_\mathcal{Z\times\mathcal{Z}}g(z)d\gamma(z, \hat{z})\,,
\end{align*}
where \eqref{marginal:first_eq} follows from the weak convergence of $\Prob_j$, \eqref{marginal:second_eq} because $\Prob_j$ is the first marginal of 
$\gamma_j^{\star}$ and finally \eqref{marginal:third_eq} because of weak convergence of $\gamma_j^{\star}$ to $\gamma$. The same holds for the second 
marginal $\Q$ of $\gamma$.
\newline
To prove (ii), consider $\gamma^{\star}$ an optimal coupling between $\Prob$ and $\Q$. Then, it holds
\begin{equation*}
\begin{split}
    &\liminf_{j\rightarrow\infty}W_c^{\epsilon}(\Prob_j, \Q_j) 
    \\
    &= \liminf_{j\rightarrow\infty} \left\{\E_{\gamma^{\star}_j}[c(z, \hat{z})] + \epsilon \mathrm{KL}(\gamma^{\star}_j|\Prob_j\times\Q_j)\right\}
    \\
    &\labelrel\geq{liminf:first_inequality} \liminf_{j\rightarrow\infty}\E_{\gamma^{\star}_j}[c(z, \hat{z})] + \epsilon\liminf_{j\rightarrow\infty}\mathrm{KL}(\gamma^{\star}_j|\Prob_j\times\Q_j) 
    \\
    &\labelrel\geq{liminf:second_inequality}\E_{\gamma}[c(z, \hat{z})] + \epsilon \mathrm{KL}(\gamma|\Prob\times\Q)
    \\
    &\labelrel\geq{liminf:third_inequality} \E_{\gamma^{\star}}[c(z, \hat{z})] + \epsilon \mathrm{KL}(\gamma^{\star}|\Prob\times\Q) = W_c^{\epsilon}(\Prob, \Q)\,,
\end{split}
\end{equation*}
where the two equalities follow from the definitions of $\gamma^{\star}_j$ and $\gamma^{\star}$, respectively. (\ref{liminf:first_inequality}) follows 
from the superadditivity of limit inferior; (\ref{liminf:second_inequality}) holds because $\E_{\gamma}[c(z, \hat{z})]$ is weakly lower semicontinuous
in $\gamma$ from \cite[Proposition 3.3]{kuhn2025distributionally} and the KL divergence term $\mathrm{KL}(\gamma|\Prob\times\Q)$ is lower semicontinuous thanks to \cite[Theorem 3.6]{polyanskiy2014}
where we have also used the fact that $\Prob_j \times \Q_j$ converges weakly to $\Prob\times\Q$ thanks to \cite[Theorem 2.8]{billingsley}. 
Finally, (\ref{liminf:third_inequality}) follows from the suboptimality of $\gamma$. This shows \eqref{eq:liminf} and thus $W_c^{\epsilon}(\Prob, \Q)$ is weakly
lower semicontinuous in $\Prob$ and $\Q$.
\subsection{Proof of Theorem \ref{thm: Sinkhorn_compactness}}
\label{appendix:theorem_compactness}
Recall that, to prove the theorem, we need to verify \cref{def:compactness}. However, it suffices to show that the $\mathcal{S}$-set is a closed subset of a compact space. To do so, we show (i) that the $\mathcal{S}$-set $\B_{\rho, \epsilon}(\Prob)$ is a subset of the compact Wasserstein ambiguity set $\mathbb{W}_{\rho^{1/p}}(\Prob)$, and (ii) that the $\mathcal{S}$-set is weakly closed. 
\newline
To prove (i), \cref{proposition: relationships} guarantees that $\B_{\rho, \epsilon}(\Prob)$ is always contained in
$\B_\rho(\Prob)$ for any $\epsilon \geq 0$ and given radius $\rho$. Moreover, 
Assumption \ref{assum: transp_cost}\ref{assum: c3} ensures that $OT_c \geq W_p^p $, which allows us to apply \cite[Proposition 2.5]{shafieezadeh2023new}
and conclude that the OT ambiguity set $\B_\rho(\Prob)$ is itself contained within the Wasserstein ambiguity set $\mathbb{W}_{\rho^{1/p}}(\Prob)$. Furthermore, we notice that the Wasserstein ball \( \mathbb{W}_{\rho^{1/p}}(\Prob) \) is weakly compact, as shown in Theorem \ref{thm:weak_compactness_Wasserstein} in Appendix \ref{appendix}. 
\newline
To prove (ii), we notice that the mapping $W_c^{\epsilon}(\cdot, \Prob)$ is lower semicontinuous with respect to the weak topology on $\mathcal{P}(\mathcal{Z})$ thanks to \cref{lemma: semicontinuity}. Hence, the ambiguity set $\B_{\rho, \epsilon}(\Prob)$ is weakly closed, being the sublevel set of a weakly lower semicontinuous mapping; see \cref{thm:rockafellar} in 
Appendix \ref{appendix}. 
We can conclude the proof because \( \B_{\rho, \epsilon}(\Prob) \) is a closed subset of the weakly compact set \( \mathbb{W}_{\rho^{1/p}}(\Prob) \), hence it is itself weakly compact by \cite[Theorem 26.2]{munkres}.
\section{System Level Synthesis}
\label{appendix:SLS}
In this section, we are going to revise the main concepts of the System-Level Synthesis (SLS) framework \cite{wang2019system}. This parametrization shifts the control synthesis problem from the direct design of the controller to the shaping of closed-loop maps from the exogenous disturbances to the state and input signals.

We start by rewriting the system's dynamics in \eqref{eq:system} over $T$ steps as
\begin{equation}
\label{eq:compact_system}
\x = Z\A\x + Z\Bmat\inp + \Emat\w, \quad
\y = \mathbf{C}\x + \mathbf{F}\w\, ,
\end{equation}
where $Z$ is the block downshift operator, that is, a block matrix with identity matrices in its first block sub-diagonal and zeros elsewhere, while $\A \doteq \text{diag}(A_0, \dots, A_{T-1}, 0_{d\times d})$, $\Bmat \doteq \text{diag}(B_0, \dots, B_{T-1}, 0_{d\times m})$, $\Emat \doteq \text{diag}(I_{d\times d}, E_0, \dots, E_{T-1})$, $\mathbf{C} \doteq \text{diag}(C_0, \dots, C_{T-1}, 0_{p\times d})$, and $\mathbf{F} \doteq \text{diag}(F_0, \dots, F_{T-1}, 0_{p\times q})$.

Let $\bm{\delta}_x = \Emat\w$ and $\bm{\delta}_y = \mathbf{F}\w$. Using \eqref{eq:compact_system} together with the feedback law \eqref{eq:compact_feedback_law}, the closed-loop behavior of the system
is characterized by the following relationships
\begin{equation}
    \label{eq:system_response}  
    \begin{bmatrix}\x\\ \inp\end{bmatrix} = \begin{bmatrix}\mathbf{\Phi}_{xx} & \mathbf{\Phi}_{xy}\\\mathbf{\Phi}_{ux} & \mathbf{\Phi}_{uy}\end{bmatrix}
    \begin{bmatrix}\bm{\delta}_x\\ \bm{\delta}_y\end{bmatrix} + \begin{bmatrix}\bm{\phi}_x\\ \bm{\phi}_u\end{bmatrix}\, ,
\end{equation}
with 
\begin{align}
\label{eq:maps}
\mathbf{\Phi}_{xx} &= (I - Z(\A +\Bmat\mathbf{KC}))^{-1}, 
&\!\!\!
\mathbf{\Phi}_{ux} &= \mathbf{KC} \mathbf{\Phi}_{xx},\nonumber
\\
\mathbf{\Phi}_{xy} &= \mathbf{\Phi}_{xx}Z\Bmat\mathbf{K},
&\!\!\!
\mathbf{\Phi}_{uy} &= \mathbf{K} + \mathbf{KC}\mathbf{\Phi}_{xx}Z\Bmat\mathbf{K},\nonumber
\\
\bm{\phi}_x &= \mathbf{\Phi}_{xx}Z\Bmat\mathbf{v}, &\!\!\!
\bm{\phi}_u &= \mathbf{KC}\bm{\phi}_x + \mathbf{v}\, .
\end{align}

With this reformulation the cost $J(\bm\pi, \w)$ with feedback law \eqref{eq:compact_feedback_law} becomes
\begin{equation}
\label{eq:costSLS}
    J(\mathbf{\Phi}, \bm\phi, \Tilde{\w}) = \bm\phi^\top D\bm\phi + \tilde{\w}^\top \mathbf{\Phi}^\top D\mathbf{\Phi}\tilde{\w}\,,
\end{equation}
with $\mathbf{\Phi} = \begin{bmatrix}\mathbf{\Phi}_{xx}\mathbf{E} & \mathbf{\Phi}_{xy}\mathbf{F}\\\mathbf{\Phi}_{ux}\mathbf{E} & \mathbf{\Phi}_{uy}\mathbf{F}\end{bmatrix}$, $\bm\phi = \begin{bmatrix}\bm{\phi}_x\\ \bm{\phi}_u\end{bmatrix}$, and $\tilde{\w} = \begin{bmatrix}
    \w\\\w
\end{bmatrix}$.

The goal is to directly optimize over the maps $\{\mathbf{\Phi}_{xx}, \mathbf{\Phi}_{xy}, \mathbf{\Phi}_{ux}, \mathbf{\Phi}_{uy}, \bm{\phi}_x, \bm{\phi}_u\}$ instead of $\bf K$ and $\bf v$. The following theorem allows us to do so.
\begin{theorem}
Given the system dynamics \eqref{eq:system} with affine output feedback law \eqref{eq:feedback law} over the time
horizon $t \in [T]$, the following are true:
\begin{enumerate}
    \item The affine subspace defined by 
    \begin{subequations}
    \label{eq:affine_subspace}
    \begin{align}
    &\begin{bmatrix}
        I - Z\A & -Z\Bmat 
    \end{bmatrix}
    \begin{bmatrix}
        \mathbf{\Phi}_{xx} & \mathbf{\Phi}_{xy}\\
        \mathbf{\Phi}_{ux} & \mathbf{\Phi}_{uy}
    \end{bmatrix} = \begin{bmatrix}
        I & 0\end{bmatrix},\label{eq:achievability1}\\
    &\begin{bmatrix}
        \mathbf{\Phi}_{xx} & \mathbf{\Phi}_{xy}\\
        \mathbf{\Phi}_{ux} & \mathbf{\Phi}_{uy}
    \end{bmatrix}
    \begin{bmatrix}
        I-Z\A\\
        -\mathbf{C}
    \end{bmatrix} = \begin{bmatrix}
        I\\
        0 \end{bmatrix}\label{eq:achievability2},
    \\
    &\begin{bmatrix}
         I - Z\A & -Z\Bmat
    \end{bmatrix}
    \begin{bmatrix}
        \bm{\phi}_{x}\\
        \bm{\phi}_{u}
    \end{bmatrix} = 0\label{eq:achievability3}\,,
\end{align}
\end{subequations}
parameterizes all possible system responses \eqref{eq:system_response}.
\item For any vectors $\{\bm{\phi}_x, \bm{\phi}_u\}$ and block lower-triangular matrices $\{\mathbf{\Phi}_{xx}, \mathbf{\Phi}_{xy}, \mathbf{\Phi}_{ux}, \mathbf{\Phi}_{uy}\}$ satisfying \eqref{eq:affine_subspace}, the affine controller in \eqref{eq:compact_feedback_law} with $\mathbf{K} = \mathbf{\Phi}_{uy} - \mathbf{\Phi}_{ux}\mathbf{\Phi}_{xx}^{-1}\mathbf{\Phi}_{xy}$ and $\mathbf{v} = \bm{\phi}_u - \mathbf{\Phi}_{uy}\mathbf{C}\bm{\phi}_x + \mathbf{\Phi}_{ux}\mathbf{\Phi}_{xx}^{-1}\mathbf{\Phi}_{xy}\mathbf{C}\bm\phi_x$ achieves the desired response.
\end{enumerate}
\end{theorem}
\begin{proof}
Proof of part 1. Let $\mathbf{K}$ be any block lower-triangular operator and $\bf v$ any real vector. We verify \eqref{eq:achievability1} by using the definitions of $\mathbf{\Phi}_{xx}$ and $\mathbf{\Phi}_{ux}$ from \eqref{eq:maps}
\begin{equation*}
\begin{aligned}
&(I\!\! -\!\! Z\A)(I\! -\! Z(\A\! +\! \mathbf{BKC}))^{-1}\!\! -\!\! Z\Bmat\mathbf{KC}(I\!\! -\!\! Z(\A\! +\! \mathbf{BKC}))^{-1}
\\
&= (I - Z\A - Z\mathbf{BKC})(I - Z(\A + \mathbf{BKC}))^{-1} = I\,,
\end{aligned}
\end{equation*}
and by using the definitions of $\mathbf{\Phi}_{xy}$ and $\mathbf{\Phi}_{uy}$
\begin{equation*}
\begin{aligned}
&(I - Z\A)(I - Z(\A + \mathbf{BKC}))^{-1}Z\Bmat\mathbf{K} - Z\Bmat\mathbf{K} -
\\
&Z\Bmat \mathbf{KC}(I - Z(\A + \mathbf{BKC}))^{-1}Z\mathbf{BK} = Z\Bmat\mathbf{K} - Z\Bmat\mathbf{K} = 0.
\end{aligned}
\end{equation*}
In a similar way we verify \eqref{eq:achievability2}. Using the definitions of $\mathbf{\Phi}_{xx}$ and $\mathbf{\Phi}_{xy}$,  one obtains
\begin{equation*}
\begin{aligned}
&(I\!\! - \!\!Z(\A\!\! +\!\! \mathbf{BKC}))^{-1}(I\!\! - \!\!Z\A)\! -\! 
(I\!\! -\!\! Z(\A\!\! +\!\! \mathbf{BKC}))^{-1}Z\Bmat\bf K\mathbf{C}
\\
&= (I - Z(\A+\mathbf{BKC}))^{-1} (I - Z(\A + \mathbf{BKC}))= I\, ,
\end{aligned}
\end{equation*}
and with the definitions of $\mathbf{\Phi}_{ux}$ and $\mathbf{\Phi}_{uy}$
\begin{equation*}
\begin{aligned}
&\mathbf{KC}(I - Z(\A + \mathbf{BKC}))^{-1}(I - Z\A) - \mathbf{KC} +
\\
&-\mathbf{KC}(I - Z(\A + \mathbf{BKC}))^{-1}Z\Bmat\mathbf{KC} = \mathbf{ KC - KC} = 0\, .
\end{aligned}
\end{equation*}
Finally, we show the validity of \eqref{eq:achievability3}. Specifically, given the definitions of $\bm{\phi}_x$ and $\bm{\phi}_u$ in \eqref{eq:maps} we can write
\begin{equation*}
\begin{aligned}
&(I - Z\A)(I - Z(\A + \mathbf{BKC}))^{-1}Z\mathbf{Bv} - Z\mathbf{Bv} -
\\
&Z\Bmat\mathbf{KC}(I - Z(\A + \mathbf{BKC}))^{-1}Z\mathbf{Bv} = Z\mathbf{Bv} - Z\mathbf{Bv} = 0\,.
\end{aligned}
\end{equation*}
Proof of part 2. We first note that $\mathbf{\Phi}_{xx}$ is a lower block-triangular matrix with identities along its diagonal, and thus
invertible. We want to show that the affine control policy $\bf u = Ky + v$ with $\mathbf{K} = \mathbf{\Phi}_{uy} - \mathbf{\Phi}_{ux}\mathbf{\Phi}_{xx}^{-1}\mathbf{\Phi}_{xy}$ and $\mathbf{v} = \bm{\phi}_u - \mathbf{\Phi}_{uy}\mathbf{C}\bm{\phi}_x + \mathbf{\Phi}_{ux}\mathbf{\Phi}_{xx}^{-1}\mathbf{\Phi}_{xy}\mathbf{C}\bm\phi_x$ achieves the desired closed-loop responses. Specifically, substituting the expression of $\bf K$ in $\mathbf{\Phi}_{xx}$ and $\mathbf{\Phi}_{ux}$ in \eqref{eq:maps} we get
\begin{align*}
&(I - Z\A -Z\Bmat\mathbf{\Phi}_{uy}\mathbf{C} 
    + Z\Bmat\mathbf{\Phi}_{ux}\mathbf{\Phi}_{xx}^{-1}\mathbf{\Phi}_{xy}\mathbf{C})^{-1} \labelrel{=}{myeq:first_equality}
\\
& (I - Z\A -Z\Bmat\mathbf{\Phi}_{ux}(I-Z\A - \mathbf{\Phi}_{xx}^{-1}\mathbf{\Phi}_{xy}\mathbf{C}))^{-1}\labelrel{=}{myeq:second_equality}  
\\
&(I - Z\A -Z\Bmat\mathbf{\Phi}_{ux} \mathbf{\Phi}_{xx}^{-1})^{-1}\labelrel{=}{myeq:third_equality} 
\\
&(((I - Z\A)\mathbf{\Phi}_{xx} -Z\Bmat\mathbf{\Phi}_{ux}) \mathbf{\Phi}_{xx}^{-1})^{-1} 
\labelrel{=}{myeq:fourth_equality} \mathbf{\Phi}_{xx}\,, \\[1em]
&\mathbf{\Phi}_{uy}\mathbf{C}\mathbf{\Phi}_{xx} 
   - \mathbf{\Phi}_{ux}\mathbf{\Phi}_{xx}^{-1}\mathbf{\Phi}_{xy}\mathbf{C}\mathbf{\Phi}_{xx}\labelrel{=}{myeq:first_equality_new}
\\
&\mathbf{\Phi}_{ux}(I-Z\A)\mathbf{\Phi}_{xx} 
- \mathbf{\Phi}_{ux}\mathbf{\Phi}_{xx}^{-1}\mathbf{\Phi}_{xy}\mathbf{C}\mathbf{\Phi}_{xx}\labelrel{=}{myeq:fifth_equality}
\\
&\mathbf{\Phi}_{ux}(I-Z\A - \mathbf{\Phi}_{xx}^{-1}\mathbf{\Phi}_{xy}\mathbf{C})\mathbf{\Phi}_{xx} \labelrel{=}{myeq:second_equality_new} \mathbf{\Phi}_{ux}\,,
\end{align*}
where in \eqref{myeq:first_equality} and \eqref{myeq:first_equality_new} we used the relation $\mathbf{\Phi}_{uy}\mathbf{C} = \mathbf{\Phi}_{ux}(I-Z\A)$ from \eqref{eq:achievability2}; in \eqref{myeq:second_equality} and \eqref{myeq:second_equality_new} we used the equation $I-Z\A - \mathbf{\Phi}_{xx}^{-1}\mathbf{\Phi}_{xy}\mathbf{C} = \mathbf{\Phi}_{xx}^{-1}$ derived from \eqref{eq:achievability2}; in \eqref{myeq:third_equality} we factored out $\mathbf{\Phi}_{xx}^{-1}$ while in \eqref{myeq:fifth_equality} $\mathbf{\Phi}_{ux}$ on the left and $\mathbf{\Phi}_{xx}$ on the right; finally, \eqref{myeq:fourth_equality} follows using $(I-Z\A)\mathbf{\Phi}_{xx} - Z\Bmat \mathbf{\Phi}_{ux} = I$ from \eqref{eq:achievability1}.

Similarly, for $\mathbf{\Phi}_{xy}$ we get:
\begin{align*}
    &\mathbf{\Phi}_{xx}Z\Bmat\mathbf{\Phi}_{uy} - \mathbf{\Phi}_{xx}Z\Bmat\mathbf{\Phi}_{ux}\mathbf{\Phi}_{xx}^{-1}\mathbf{\Phi}_{xy}\labelrel={eq4:first}
    \\
    &\mathbf{\Phi}_{xx}(I-Z\A - Z\Bmat\mathbf{\Phi}_{ux}\mathbf{\Phi}_{xx}^{-1})\mathbf{\Phi}_{xy} \labelrel={eq4:second} \mathbf{\Phi}_{xy}\,,
\end{align*}
where in \eqref{eq4:first} we used the relation $Z\Bmat \mathbf{\Phi}_{uy} = (I-Z\A)\mathbf{\Phi}_{xy}$ from \eqref{eq:achievability1} and we factored out $\mathbf{\Phi}_{xx}$ on the left and $\mathbf{\Phi}_{xy}$ on the right; \eqref{eq4:second} instead uses $I-Z\A - Z\Bmat\mathbf{\Phi}_{ux}\mathbf{\Phi}_{xx}^{-1} =  \mathbf{\Phi}_{xx}^{-1}$ derived from \eqref{eq:achievability1}. 

We conclude the linear part by verifying $\mathbf{\Phi}_{uy}$. Specifically:
\begin{align*}
\label{eq:Phi_uy}
&\mathbf{\Phi}_{uy} -\mathbf{\Phi}_{ux}\mathbf{\Phi}_{xx}^{-1}\mathbf{\Phi}_{xy} + 
\\
& (\mathbf{\Phi}_{uy}\mathbf{C} - \mathbf{\Phi}_{ux}\mathbf{\Phi}_{xx}^{-1}\mathbf{\Phi}_{xy}\mathbf{C})\mathbf{\Phi}_{xx} Z\Bmat\mathbf{\Phi}_{uy} -
\\
&(\mathbf{\Phi}_{uy}\mathbf{C} - \mathbf{\Phi}_{ux}\mathbf{\Phi}_{xx}^{-1}\mathbf{\Phi}_{xy}\mathbf{C})\mathbf{\Phi}_{xx}Z\Bmat
\mathbf{\Phi}_{ux}\mathbf{\Phi}_{xx}^{-1}\mathbf{\Phi}_{xy} \labelrel{=}{eq5:first} 
\\
& \mathbf{\Phi}_{uy} -\mathbf{\Phi}_{ux}\mathbf{\Phi}_{xx}^{-1}\mathbf{\Phi}_{xy} +
\\
&\mathbf{\Phi}_{ux}(I\!\!-\!\!Z\A\!\! -\! \mathbf{\Phi}_{xx}^{-1}\mathbf{\Phi}_{xy}\mathbf{C})\mathbf{\Phi}_{xx}(I\!\!-\!\!Z\A\!\! -\! Z\Bmat \mathbf{\Phi}_{ux}\mathbf{\Phi}_{xx}^{-1})\mathbf{\Phi}_{xy}\labelrel{=}{eq5:second} 
\\
&\mathbf{\Phi}_{uy} -\mathbf{\Phi}_{ux}\mathbf{\Phi}_{xx}^{-1}\mathbf{\Phi}_{xy} +\mathbf{\Phi}_{ux}\mathbf{\Phi}_{xx}^{-1}\mathbf{\Phi}_{xy} = \mathbf{\Phi}_{uy}\,,
\end{align*}
where in \eqref{eq5:first} we used $\mathbf{\Phi}_{uy}\mathbf{C} = \mathbf{\Phi}_{ux}(I-Z\A)$ from \eqref{eq:achievability2} and $Z\Bmat \mathbf{\Phi}_{uy} = (I-Z\A)\mathbf{\Phi}_{xy}$ from \eqref{eq:achievability1}; while in \eqref{eq5:second} we used the relations $I-Z\A - Z\Bmat\mathbf{\Phi}_{ux}\mathbf{\Phi}_{xx}^{-1} =  \mathbf{\Phi}_{xx}^{-1}$ derived from \eqref{eq:achievability1} and $I-Z\A - \mathbf{\Phi}_{xx}^{-1}\mathbf{\Phi}_{xy}\mathbf{C} = \mathbf{\Phi}_{xx}^{-1}$ derived from \eqref{eq:achievability2}.
We finish the proof by substituting $\mathbf{v} = \bm{\phi}_u - \mathbf{\Phi}_{uy}\mathbf{C}\bm{\phi}_x + \mathbf{\Phi}_{ux}\mathbf{\Phi}_{xx}^{-1}\mathbf{\Phi}_{xy}\mathbf{C}\bm\phi_x$ in \eqref{eq:maps}:
\begin{align*}
    &\mathbf{\Phi}_{xx}Z\Bmat\bm\phi_x - \mathbf{\Phi}_{xx}Z\Bmat(\mathbf{\Phi}_{uy}\mathbf{C}\bm\phi_x + \mathbf{\Phi}_{ux}\mathbf{\Phi}_{xx}^{-1}\mathbf{\Phi}_{xy}\mathbf{C})\bm\phi_x\labelrel{=}{affine:first} 
    \\
    &\mathbf{\Phi}_{xx}(I\!-\!Z\A)\bm\phi_x\! -\! \mathbf{\Phi}_{xx}Z\Bmat\mathbf{\Phi}_{ux}(I\!-\!Z\A\!-\!\mathbf{\Phi}_{xx}^{-1}\mathbf{\Phi}_{xy}\mathbf{C})\bm\phi_x\labelrel{=}{affine:second}
    \\
    &\mathbf{\Phi}_{xx}(I - Z\A - Z\Bmat\mathbf{\Phi}_{ux}\mathbf{\Phi}_{xx}^{-1})\bm\phi_x \labelrel{=}{affine:third} \bm\phi_x\,,
    \\[1em] 
    &\begin{multlined}[t]
        \mathbf{\Phi}_{uy}\mathbf{C}\bm{\phi}_x - \mathbf{\Phi}_{ux}\mathbf{\Phi}_{xx}^{-1}\mathbf{\Phi}_{xy}\mathbf{C}\bm\phi_x + \bm{\phi}_u +
        \\
        - \mathbf{\Phi}_{uy}\mathbf{C}\bm{\phi}_x + \mathbf{\Phi}_{ux}\mathbf{\Phi}_{xx}^{-1}\mathbf{\Phi}_{xy}\mathbf{C}\bm\phi_x = \bm{\phi}_u\,,
    \end{multlined}
\end{align*}
where in \eqref{affine:first} we used $(I-Z\A)\bm\phi_x = Z\Bmat\bm\phi_u$ from \eqref{eq:achievability3} and $\mathbf{\Phi}_{uy}\mathbf{C} = \mathbf{\Phi}_{ux}(I-Z\A)$ from \eqref{eq:achievability2}; in \eqref{affine:second} we factored out $\mathbf{\Phi}_{xx}$ on the left and $\bm\phi_x$ on the right and used $I-Z\A - \mathbf{\Phi}_{xx}^{-1}\mathbf{\Phi}_{xy}\mathbf{C} = \mathbf{\Phi}_{xx}^{-1}$ derived from \eqref{eq:achievability2}; finally, \eqref{affine:third} uses $I-Z\A - Z\Bmat\mathbf{\Phi}_{ux}\mathbf{\Phi}_{xx}^{-1} =  \mathbf{\Phi}_{xx}^{-1}$ derived from \eqref{eq:achievability1}.
\end{proof}

\section{Proofs for \texorpdfstring{\cref{sec:optimal_control}}{Section IV}}
\subsection{Proof of Proposition \ref{proposition: DR_Sinkhorn_CVaR}}
\label{appendix: DR_Sinkhorn_CVaR}
To prove the Proposition we need a technical result that is given in the following lemma.
\begin{lemma}
\label{lemma: technical}
Consider $J\in\N$ functions $f_1, \ldots, f_J:\R^d\rightarrow\R$ and a scalar function $g:\R\rightarrow\R$. Then, $\forall x\in\R^d$,
\begin{equation*}
    g(f_j(x)) \leq 0,\ \forall j\in[J] \implies g\left(\max_{i\in[J]} f_i(x)\right) \leq 0\,.
\end{equation*}
Moreover, if the function $g$ is monotonically non-decreasing, $\forall x\in\R^d$ we have
\begin{equation*}
    g\left(\max_{i\in[J]} f_i(x)\right) \leq 0 \implies g(f_j(x)) \leq 0,\ \forall j\in[J]\,.
\end{equation*}
\end{lemma}
\begin{proof}
The first implication is trivial. We now prove the second one. For any $x\in\R^d$, $\max_{i\in[J]} f_i(x) \geq f_j(x),\ \forall j\in[J]$. Therefore, if
we consider the non-decreasing function $g$ we have $g\left(\max_{i\in[J]} f_i(x)\right) \geq g(f_j(x)),\ \forall j\in[J]$ because $g$ preserves the ordering.
Then, $g\left(\max_{i\in[J]} f_i(x)\right) \leq 0$ implies $g(f_j(x)) \leq 0,\ \forall j\in[J]$. This concludes the proof.
\end{proof}
Using the CVaR definition (\ref{eq:CVaR}) and defining $\alpha_j := a_j/\gamma$ and $\beta_j(\tau) := (b_j + \gamma\tau -\tau)/\gamma$ for $j\in[J+1]$ with 
$a_{J+1} = 0$ and $b_{J+1} = \tau$, the left hand side of the constraint in (\ref{eq:DR_Sinkhorn_CVaR}) can be written as
\begin{equation*}
     \sup_{\Q\in\B_{\rho,\epsilon}(\Prob)}\ \inf_{\tau\in\R}\ \E_\Q\left[\max_{j\in[J+1]} \alpha_j^\top x + \beta_j(\tau)\right]\,.
\end{equation*}
From Theorem \ref{thm: Sinkhorn_compactness} and Proposition \ref{proposition:convexity} we know that the $\mathcal{S}$-set is weakly compact and 
convex. The argument of the expectation is an affine function of $\tau$ while the expectation is linear in the probability $\Q$ hence the 
optimization problem satisfies the assumptions of Sion's minmax theorem \cite{Sion1958}. Therefore, we can interchange the order of 
maximization and minimization without affecting the optimality of the solution obtaining
\begin{equation*}
   \inf_{\tau\in\R}\ \sup_{\Q\in\B_{\rho,\epsilon}(\Prob)}\ \E_\Q\left[\max_{j\in[J+1]} \alpha_j^\top x + \beta_j(\tau)\right]\,.
\end{equation*}
We now focus on the inner supremum and derive a dual problem (for simplicity we drop the dependency on $\tau$). We consider the strong dual reformulation (\ref{eq:dual}) with 
quadratic transport cost and piecewise affine loss function. Condition \eqref{eq:feasibility_CVaR} can be derived analogously as in \cite[Lemma 1]{cescon2025data}. In this setup, the dual program becomes
\begin{align*} 
    \inf\ &\sigma \rho + \frac{1}{n} \sum_{i=1}^n \tilde{s}_i\\
    &\text{s.t.}\quad \forall i\in [n]:\ \tilde{s}_i\in\R, \sigma\in\R_+\\
    &\sigma\epsilon \log\E_{z \sim \nu}\left[e^{\max_{j\in[J+1]}\ (\alpha_j^\top z + \beta_j - \sigma\|\hat{\xi}_i-z\|^2)/(\sigma\epsilon)}\right]\! \leq\! \tilde{s}_i\,,
\end{align*}
where the term $\|\hat{\xi}_i-z\|^2$ does not depend on $j$ and can be taken inside the maximization.
Using \cref{lemma: technical} we can argue that the maximization over the $J+1$ pieces can be dropped if the inequality constraint is satisfied 
for all $j\in[J+1]$. This is true since it can be readily verified that the function $g(\cdot) = \log\E_{\nu}[e^{(\cdot)}]$ is monotonically non-decreasing.
Therefore, we can solve the expectation over a single affine piece which rewrites as
\begin{equation*}
\begin{aligned}
&\E_{z \sim \nu}\left[e^{(\alpha_j^\top z + \beta_j - \sigma\|\hat{\xi}_i - z\|^2)/(\sigma\epsilon)}\right]=
\\
&
\begin{multlined}[t]
    C_d^{-1}\int_{\R^d}\exp\biggl\{\left(\alpha_j^\top z + \beta_j - \sigma(z - \hat{\xi}_i)^{\top}(z - \hat{\xi}_i)\right)/(\sigma\epsilon)
    \\
    -\frac{1}{2}(z-m)^\top \Sigma^{-1}(z-m)\biggr\}\dd\lambda^d(z)=
\end{multlined}
\\
&
\begin{aligned}
   \kappa\exp\biggl\{\! \biggl[\star^{\top}\! &\left(4\sigma I + 2\sigma\epsilon\Sigma^{-1}\right)^{-1}\!
    \left(\alpha_j + 2\sigma\hat{\xi}_i + \sigma\epsilon\Sigma^{-1}m\right)
    \\
    & + \beta_j- \sigma \|\hat{\xi}_i\|^2 - \frac{\sigma\epsilon}{2}\|m\|^2_{\Sigma^{-1}}\biggr]/(\sigma\epsilon)\biggr\}\,,
    \end{aligned}
\end{aligned}
\end{equation*}
where $\kappa = \frac{\epsilon^{d/2}}{\sqrt{|2\Sigma + \epsilon I|}}$. 

Considering the logarithm of the previous expression scaled by $\sigma\epsilon$ we get
\begin{equation*}
    \begin{aligned}
    \frac{\sigma \epsilon d}{2}\log\epsilon - \frac{\sigma\epsilon}{2}\log|2\Sigma + \epsilon I| + \beta_j- \sigma \|\hat{\xi}_i\|^2 
    - \frac{\sigma\epsilon}{2}\|m\|^2_{\Sigma^{-1}}
    \\
    + \star^{\top}\! \left(4\sigma I + 2\sigma\epsilon\Sigma^{-1}\right)^{-1}\! \left(\alpha_j + 2\sigma\hat{\xi}_i + \sigma\epsilon\Sigma^{-1}m\right)\,.
    \end{aligned}
\end{equation*}

Adopting the change of variables $s_i = \tilde{s}_i - \frac{\gamma-1}{\gamma}\tau$ and using the Schur's complement, the optimization program rewrites
$\forall j\in[J+1]$ as 
\begin{align*} 
    &\inf\ \sigma \rho + \frac{\gamma-1}{\gamma}\tau + \frac{1}{n} \sum_{i=1}^n s_i
    \\
    &\ \text{s.t.}\ \forall i\in [n]:\ s_i\in\R,\sigma\in\R_+
    \\
    &\begin{bmatrix}
        4\sigma I + \sigma\epsilon\Sigma^{-1} &\frac{a_j}{\gamma} + 2\sigma \hat{\xi}_i + \sigma\epsilon\Sigma^{-1}m
        \\
        \star & s_i - \zeta - \frac{b_j}{\gamma} + \sigma\|\hat{\xi}_i\|^2 + \frac{\sigma\epsilon}{2}\|m\|^2_{\Sigma^{-1}}
    \end{bmatrix}\! \succeq\! 0\,,
\end{align*}
with $\zeta = \frac{\sigma\epsilon d}{2}\log\epsilon - \frac{\sigma\epsilon}{2}\log|2\Sigma + \epsilon I|$.
Finally, the result of Proposition \ref{proposition: DR_Sinkhorn_CVaR} follows noticing that $\inf \cdot \leq 0$ constraints are equivalent to existance constraints.
\subsection{Detailed Calculations for \texorpdfstring{\cref{remark: Wasserstein_CVaR}}{Remark 1}}
\label{appendix:remark_proof}
From \cite[Theorem 8]{kuhn2019wasserstein}, we can derive the expression for the DR CVaR constraint (\ref{eq:DR_Sinkhorn_CVaR}) 
when the ambiguity set is defined as a Wasserstein ball, i.e. when $\epsilon\rightarrow 0$. Indeed, the constraint is equivalent to the following convex set:
\begin{equation}
    \label{eq:Wasserstein_CVaR}
    \begin{aligned}
    &\forall i \in [n], \forall j\in[J+1]:\\
    &\left\{
        \begin{aligned}
        & \tau\in\R, \sigma\in\R_+, s_i\in\R\\ 
        &\sigma \rho + \frac{\gamma - 1}{\gamma}\tau+\frac{1}{n} \sum_{i=1}^n s_i \leq 0\\
        &\frac{1}{\gamma}(b_j + a_j^\top \hat{\xi}_i) + \frac{\|a_j\|^2}{4\sigma\gamma^2} \leq s_i\,,
        \end{aligned}
    \right.
    \end{aligned}
\end{equation}
with $a_{J+1} = 0$ and $b_{J+1} = \tau$.

If we focus on the convex set of Proposition \ref{proposition: DR_Sinkhorn_CVaR} when $\epsilon\rightarrow 0$, the first inequality matches the one in \eqref{eq:Wasserstein_CVaR} while the LMI constraint becomes
\begin{equation*}
    \begin{bmatrix}
        4\sigma I  &\frac{a_j}{\gamma} + 2\sigma \hat{\xi}_i\\
        \star & s_i - \frac{b_j}{\gamma} + \sigma\|\hat{\xi}_i\|^2
    \end{bmatrix} \succeq 0\,.
\end{equation*}
With Schur's complement we can rewrite the LMI as
\begin{equation*}
    s_i - \frac{b_j}{\gamma} + \sigma\|\hat{\xi}_i\|^2 - \frac{1}{4\sigma}\left(\frac{a_j}{\gamma} + 2\sigma\hat{\xi}_i\right)^\top
    \left(\frac{a_j}{\gamma} + 2\sigma\hat{\xi}_i\right) \geq 0\,,
\end{equation*}
and, after some simple algebraic manipulations, we obtain
\begin{equation*}
    s_i - \frac{b_j}{\gamma} - \frac{a_j^\top}{\gamma}\hat{\xi}_i - \frac{\|a_j\|^2}{4\sigma\gamma^2} \geq 0\,,
\end{equation*}
which is equivalent to the second inequality in \eqref{eq:Wasserstein_CVaR}.
\subsection{Dual problem for quadratic cost and loss function}
\label{appendix:lemma}
The following Lemma was derived in \cite{cescon2025data} and it is used to prove \cref{thm:main}
\begin{lemma}
\label{thm: duality}
Under the same assumptions of Theorem \ref{thm:main}, if $\ell(z) = z^\top Qz + 2q^\top z$ with $Q\in\Symm^d$, $q\in\R^d$ then problem (\ref{eq:worst-case risk}) is feasible if and only if condition \eqref{eq:feasibility_condition_control} holds
and the optimal value of (\ref{eq:worst-case risk}) coincides with the optimal value of the following convex optimization problem
\begin{align}
\label{eq:convex program}
\inf\ &\lambda\rho + \frac{1}{n}\sum_{i=1}^n s_i
\\
\textrm{s.t.}\ & s_i\in\R, \lambda\in\R_+, \lambda\left(I+\frac{\epsilon}{2}\Sigma^{-1}\right) \succ Q, \quad \forall i\in[n]
\nonumber\\
& \frac{\lambda\epsilon d}{2}\log\left(\!\frac{\lambda\epsilon}{2}\!\right)\! -\! \frac{\lambda\epsilon}{2}\log|\Sigma|\!-\! \frac{\lambda\epsilon}{2}\log\left|\lambda \left(I \!+\! \frac{\epsilon}{2}\Sigma^{-1}\!\right)\! -\! Q\right|+
\nonumber
\\
&+ \star^{\top}\left(\lambda\left(I\! +\!\frac{\epsilon}{2}\Sigma^{-1}\right)\!-\!Q\right)^{-1}\!\left(q+\lambda\left(\hat{\xi}_i+\frac{\epsilon}{2}\Sigma^{-1}m\right)\right)+
\nonumber
\\
&- \lambda \|\hat{\xi}_i\|^2 - \frac{\lambda\epsilon}{2}\|m\|^2_{\Sigma^{-1}}\leq s_i\,.\nonumber
\end{align} 
\end{lemma}
\subsection{Proof of Proposition \ref{prop:convexity_problem}}
\label{app:proof_convexity}
The constraints \eqref{eq:M>0}, \eqref{eq:LMI2}, \eqref{eq:CVaR2} $\forall i\in[N]$, and \eqref{eq:LMI1 Schur} are linear matrix inequalities and the achievability constraints on $\mathbf{\Phi}, \bm\phi$ are affine; hence they are convex. Moreover, \eqref{eq:CVaR1} is affine in the optimization variables hence convex. We focus now on the non-linear constraint \eqref{eq:logdet inequality}. This constraint is linear in $s_i$ and $\zeta_i$. We proceed to show that the non-linear part $\frac{\lambda\epsilon s}{2}\log\left(\frac{\lambda\epsilon}{2}\right) - \frac{\lambda\epsilon}{2}\log\left|\lambda \left(I + \frac{\epsilon}{2}\Sigma^{-1}\right) - Q\right|$ of \eqref{eq:logdet inequality} is jointly convex in $(\lambda, Q)$ for every $\epsilon\geq 0$. To do so, we define the functions $h:\mathbb{S}^s\rightarrow\R$ and $T:\mathbb{S}^s\rightarrow\mathbb{S}^s$ as follows:
\begin{equation*}
h(Q) = \frac{\epsilon s}{2}\log\left(\frac{\epsilon}{2}\right)- \frac{\epsilon}{2}\log |Q|\,, \quad T(Q) = I + \frac{\epsilon}{2}\Sigma^{-1} - Q\,.
\end{equation*}
We note that $h$ is convex in $Q$ because the log-determinant of a matrix is a concave function \cite{boyd2004convex}. Similarly, $T$ is affine and therefore convex in $Q$. Hence, the function $g:\mathbb{S}^s\rightarrow\R$, 
\[g(Q) = h(T(Q)) = \frac{\epsilon}{2}\log\left(\frac{(\frac{\epsilon}{2})^s}{\left|I + \frac{\epsilon}{2}\Sigma^{-1} - Q\right|}\right)\,,\]
is convex because composition of an affine and a convex function. We then note that the non-linear part of \eqref{eq:logdet inequality} can be rewritten as
\begin{equation*}
    f(\lambda, Q) = \frac{\lambda\epsilon}{2}\log\left(\frac{\left(\frac{\lambda\epsilon}{2}\right)^s}{\left|\lambda \left(I + \frac{\epsilon}{2}\Sigma^{-1}\right) - Q\right|}\right)\,,
\end{equation*}
and that $f(\lambda, Q) = \lambda g(Q/\lambda)$ is the perspective of the function $g$, see \cite[Section 3.2.6]{boyd2004convex}. Since the perspective of a convex function is also convex, the proof is concluded.

\end{document}